\documentclass[11pt]{article}
\usepackage{style}
\usepackage[left=1in,bottom=1in,top=1in,right=1in]{geometry}
\hypersetup{%
	colorlinks   = true, %
	urlcolor     = blue, %
	linkcolor    = blue, %
	citecolor    = blue,  %
}

\usepackage[dvipsnames]{xcolor}
\usepackage{enumitem}

\title{Pareto-optimal Non-uniform Language Generation}

\author{Moses Charikar\thanks{Stanford University. Email: \texttt{moses@cs.stanford.edu.}}
\and
Chirag Pabbaraju\thanks{Stanford University. Email: \texttt{cpabbara@cs.stanford.edu.}}}
\date{\today}

\graphicspath{{img/}}

\begin{document}

\maketitle
\begin{abstract}
    \citet*{kleinberg2024language} 
    recently proposed an interesting model for language generation in the limit:
    Given a countable collection of languages, and an adversary enumerating the strings of some language $L$ from the collection, %
    the objective is to generate \textit{new} strings from the target language, such that all strings generated beyond some finite time are valid.
    \citet*{li2024generation} and 
    \citet*{charikar2024exploring} 
    showed strong
    \textit{non-uniform} generation guarantees in this model, 
    giving algorithms that generate new valid strings from $L$ after seeing a number of distinct input strings $t(L)$ that depends only on $L$ (and the collection), but \textit{not} the %
    enumeration order.
    However, for both these works, %
    the language-wise \textit{generation times} $t(L)$ of the algorithm can be strictly sub-optimal.

    In this work, we study %
    \textit{Pareto-optimality}
    of non-uniform language generation in the limit. 
    We propose an algorithm, whose generation times $t^\star(L)$ are (almost) Pareto-optimal: any other algorithm whose generation time for some language $L$ is strictly smaller than $t^\star(L)$, \textit{must satisfy} that its generation time for some \textit{other} language $L'$ is strictly worse than $t^\star(L')$. %
    Pareto-optimality is essentially the best that one can achieve for non-uniform generation. %
    Our algorithmic framework conveniently adapts to further give Pareto-optimal non-uniform generation algorithms in the practically motivated settings of \textit{noisy} as well as \textit{representative} generation. 
\end{abstract}
\newpage
\section{Introduction}
\label{sec:intro}

The phenomenal success of large language models in generating coherent language motivates the study of concrete theoretical models to explain their working. Synthesizing a formal model that is tractable to analyze, but also representative enough, is challenging, given the extremely complex architecture of these large language models. Nevertheless, the foundational work of \cite{gold1967language} and \cite{angluin1979finding, angluin1980inductive} on language identification provides an excellent starting point, and motivated the recent formulation of language generation in the limit by \cite{kleinberg2024language}. The setup is as follows: there is a countable collection $\mcC=(L_1,L_2,L_3,\ldots)$ of languages (where each language is an infinite subset of strings from a common universe). An adversary chooses a target language $L$, and lists the strings in $L$ in an online fashion according to an ordering of their choice. This corresponds to the ``training data'' that is input to a generative language model which receives a string as input and produces a string as output at every time step.
The goal of the learning algorithm is \textit{generation in the limit}, i.e.,~to guarantee that beyond some finite time, the algorithm \textit{only} generates new strings from $L$ that have not yet appeared in the input. Under this model, \cite{kleinberg2024language} establish a surprisingly general positive result: there exists an algorithm that achieves the guarantee of generation in the limit for \textit{any} countable collection $\mcC$. 

However, the time at which the algorithm correctly starts generating can be very sensitive to the precise order in which the language is being enumerated as input; as such, the adversary might be able to arbitrarily delay the time of correct generation, which may be undesirable. This motivates the notion of \textit{non-uniform} language generation introduced in the work of \cite{li2024generation}. Non-uniform language generation requires that for every language $L$ that the adversary may chose from the collection, the algorithm starts generating validly as soon as the input comprises of $t(L)$ many distinct strings, \textit{independent} of what these strings are, and what order they are enumerated in. This is a considerably stronger guarantee compared to vanilla generation in the limit, %
and a priori, it would appear that only a strict subset of countable collections may admit non-uniform generation. Nevertheless, as shown independently by both \cite{li2024generation} and \cite{charikar2024exploring}, non-uniform generation in the limit also turns out to be possible for every countable collection!

Given these strong positive results, a natural direction that has %
not been investigated yet is to determine the \textit{optimal} non-uniform generation algorithm, which requires the least possible number of inputs for \textit{every language simultaneously} in the collection. 
However, for any language $L$, there is an algorithm that achieves $t(L)=1$.
Such an algorithm simply generates strings from $L$ 
at the outset,
until %
it sees a string outside $L$.
The algorithms of \cite{li2024generation} and \cite{charikar2024exploring} (see \Cref{appsec:example})
can do this for any language $L$. 
Since the optimal generation time for any language $L$ is $1$,
achieving optimality simultaneously for all languages $L$ is not feasible.

Nevertheless, in this work, we initiate a study of \textit{Pareto-optimal} non-uniform language generation. 
We say that a sequence of generation times $t(L_1), t(L_2),\ldots$ for languages in the collection is Pareto-optimal, if it satisfies the following property: any algorithm whose generation time for some $L_i$ is smaller than $t(L_i)$ \textit{must} satisfy that its generation time for some other language $L_j$ is larger than $t(L_j)$.\footnote{This is a slight abuse of notation, since our definition of Pareto-optimal sequence of generation times does not require an algorithm which \textit{achieves} this sequence. Technically, this is a lower bound on the Pareto front.}
As our first contribution, we give a constructive process to obtain a canonical Pareto-optimal sequence of generation times for any given collection. The process resembles the textbook insertion sort algorithm, albeit with a carefully chosen comparator, and incrementally computes an ordering of the first $n$ languages in the collection, for increasing $n$. Maintaining the ordering given by the comparator ensures that the generation time for every language is determined Pareto-optimally.

The previous non-uniform generation algorithms of \cite{li2024generation} and \cite{charikar2024exploring} do not achieve Pareto-optimality %
(see \Cref{example:lrt-cp-suboptimal}). Taking inspiration from our constructive process, we next give a non-uniform generation algorithm that is \textit{almost} Pareto-optimal, in the following sense: the generation times of the algorithm can be made to match the generation times in the canonical Pareto-optimal sequence constructed above for an arbitrarily large (but finite) prefix of languages in the collection. 

\begin{theorem}[Almost Pareto-optimal Non-Uniform Language Generation]
    \label{thm:informal-almost-Pareto-optimal-non-uniform-language-generation}
    Given a collection $\mcC=(L_1,L_2,\ldots)$, for any $n < \infty$, there exists an algorithm that non-uniformly generates from $\mcC$, and its sequence of generation times $t(L_1),t(L_2),\ldots$ for languages in $\mcC$ satisfies the following: any other algorithm whose generation time for some language $L_i$ for $i \le n$ is smaller than $t(L_i)$, must satisfy that its generation time for some other language $L_j$ (also for $j \le n$) is larger than $t(L_j)$.
\end{theorem}

We also identify a technical condition on a collection that is sufficient to guarantee that the algorithm from the theorem above attains Pareto-optimality for the \textit{entire} collection (\Cref{thm:Pareto-optimality-guarantee}); that is, its sequence of generation times \textit{entirely} matches the Pareto-optimal sequence given by our constructive process, as opposed to matching it on just a finite prefix. Could the sufficient condition also be necessary? %
Is it possible to
achieve Pareto-optimal generation times for all countable collections? Towards this, we identify two simple collections that do not satisfy our sufficient condition (\Cref{claim:Pareto-impossibility-non-uniform} and \Cref{claim:Pareto-impossibility-uniform}) %
for which
it is \textit{impossible} for any algorithm to attain a Pareto-optimal sequence of generation times for the entire collection. 
We leave open the tantalizing question of obtaining an exact characterization of collections that admit Pareto-optimal non-uniform generation for future work.

Finally, we demonstrate that our framework above provides a \textit{generic blueprint} for obtaining Pareto-optimal algorithms, even under other specialized settings requiring stronger guarantees than standard non-uniform generation. Namely, we study the recently introduced settings of \textit{noisy} \citep*{raman2025generation} and \textit{representative} \citep*{peale2025representative} generation. 
We show how suitably modifying the comparator in our insertion sorting process makes both these settings fit into our framework, allowing us to obtain almost Pareto-optimal non-uniform generation algorithms for both these settings. Lastly, we derive sufficient conditions as above for when our algorithms attain exact Pareto-optimality in these settings.

\section{Preliminaries}
\label{sec:prelims}

We first describe the setup of language generation in the limit introduced by \citet{kleinberg2024language}. In this model, there is an underlying countable collection of languages, where each language is a countably infinite subset of a countably infinite universe $U$. %
We use the notation $(L_1,L_2,\ldots)$ to denote the languages in a countable collection $\mcC$; we use this notation, in lieu of curly brackets $\{\}$, to emphasize that the ordering is important, and that there might be duplicates of a language at multiple indices.%

An \textit{enumeration} of a language $L$ comprises of a sequence $x_1,x_2,\ldots$ of strings (repetitions allowed) such that $x_t \in L$ for every $t \ge 1$, and furthermore, for every $x \in L$, there exists a finite $t$ such that $x_t = x$. The generation in the limit game operates as follows: an adversary chooses a language from the known collection $\mcC$, and proceeds to enumerate it in a worst-case fashion. At each time step $t$, after having observed the sequence $x_1,\ldots,x_t$ enumerated so far by the adversary, the task of a generating algorithm is to generate a string $z_t$. Note that the input does not comprise of any strings outside the chosen language, and the algorithm gets no feedback about the string it generates at any time step. We denote the set of \textit{distinct} strings in the input sequence $x_1,\ldots,x_t$ enumerated up until time $t$ by $S_t$. The success criterion for the algorithm is defined as follows:

\begin{definition}[Generation in the Limit \citep{kleinberg2024language}]
    \label{def:language-generation-in-the-limit}
    A generating algorithm algorithm $\mcG$ generates in the limit from a collection $\mcC=(L_1,L_2,\ldots)$ if for every $L \in \mcC$ and every enumeration $\sigma$ of $L$, there exists a finite $t(L, \sigma)$\footnote{We suppress the implicit dependence on the underlying collection $\mcC$ for convenience.} such that for every $t \ge t(L, \sigma)$, the string $z_t$ generated by the algorithm at time step $t$ belongs to $L \setminus S_t$.
\end{definition}

The collection above plays the role of the ``function class'' from which a generating algorithm is trying to learn. For example, we can imagine that each language in the collection represents the set of strings that a language model may generate corresponding to a given configuration of its parameters, and as more inputs get revealed with time, the model tunes its parameters, simultaneously moving around in the collection.

The general positive result of \cite{kleinberg2024language} shows that generation in the limit can achieved for \textit{every} countable collection of languages. Observe however that the time for valid generation in the definition above, namely $t(L,\sigma)$, is a function of both the target language $L$ as well as the enumeration order $\sigma$. The notion of non-uniform language generation tries to get rid of the dependence on the enumeration order $\sigma$ in the success time $t(L, \sigma)$.

\begin{definition}[Non-uniform Generation \citep{li2024generation}]
    \label{def:non-uniform-generation}
    A generating algorithm algorithm $\mcG$ non-uniformly generates in the limit from a collection $\mcC=(L_1,L_2,\ldots)$ if for every $L \in \mcC$, there exists a $t(L)$ such that for every enumeration $\sigma$ of $L$ presented to $\mcG$, the string $z_t$ generated by the algorithm at time step $t$ belongs to $L \setminus S_t$ for all $t$ satisfying $|S_t| \ge t(L)$.
\end{definition}

As it turns out, the stronger guarantee of non-uniform generation in the limit can also be achieved for every countable collection of languages, as shown by \citet{li2024generation} and \citet{charikar2024exploring}. It helps to think about the language-dependent success time $t(L)$ as measuring how easy it is to discern $L$ from other languages in the collection: the larger $t(L)$ is, the larger the amount of resources---number of inputs, computation---that an algorithm must invest on language $L$ in order to generate validly from it. We can then ask: what is the minimal amount of resources that can be assigned to every language in the collection? Towards this, we propose the following definition of Pareto-optimality:

\begin{definition}[Pareto-optimality]
    \label{def:Pareto-optimal-generation-times}
    A sequence $t(L_1), t(L_2),\ldots$ for languages in $\mcC$ is Pareto-optimal if any algorithm $\mcG$ that non-uniformly generates from $\mcC$ and satisfies that $t_\mcG(L_i) < t(L_i)$ for some $i$, must satisfy that its generation time for some other $L_j$ satisfies $t_\mcG(L_j) > t(L_j)$
\end{definition}

In our study of this guarantee, the notion of \textit{finite intersections} between languages will play a key role. A subcollection $\mcC'$ of $\mcC$ has finite intersection if $|\cap_{L \in \mcC'}L| < \infty$. The subcollection $\mcC'$ of $\mcC$ that has \textit{largest} finite intersection satisfies
\begin{align*}
    \mcC' \in \argmax_{\mcC''}\{|\cap_{L \in \mcC''}L|: \mcC'' \subseteq \mcC, |\cap_{L \in \mcC''}L|<\infty\}.
\end{align*}

\subsection{Prior Guarantees}
\label{appsec:example}

We now formally review the non-uniform generation guarantees of the algorithms given by \cite{li2024generation} and \cite{charikar2024exploring}. We will specify the generation time required by these algorithms for each language $L_i$ in the collection $\mcC=(L_1,L_2,\ldots)$.

\paragraph{Guarantee of \cite{li2024generation}.} Define the \textit{closure dimension} of a collection $\mcC$ to be the size of the largest finite intersection of a subcollection of $\mcC$:\footnote{\cite{li2024generation} originally define the closure dimension as the size of the largest set $T$ such that $|\cap_{L \in \mcC, T \subseteq L}L| < \infty$; the definition stated here is equivalent.}
\begin{align}
    \label{def:closure-dim}
    d := \max_{\mcC'}\{|\cap_{L \in \mcC'}L|: \mcC' \subseteq \mcC, |\cap_{L \in \mcC'}L|<\infty\}\footnotemark,
\end{align}
\footnotetext{For concreteness, we define the max over an empty set to be 0, and the arg max to be null.}
For every language $L_i$, consider the closure dimension $d_i$ of the prefix collection $(L_1,\ldots,L_i)$.
Observe that the sequence $\{d_i\}_{i \ge 1}$ is non-decreasing. The guarantee for the algorithm given by 
\cite{li2024generation} is as follows: If $\lim_{i \to \infty}d_i = \infty$, then the generation time for each language is $t(L_i)=d_i+1$. If $\lim_{i \to \infty}d_i=c < \infty$, then $t(L_i)=\max(i, c+1)$.

\paragraph{Guarantee of \cite{charikar2024exploring}.} For any $L_i$, consider the largest finite intersection of a subcollection of $(L_1,\dots,L_i)$ \textit{that contains} $L_i$, called the \textit{non-uniform complexity} of $L_i$:
    \begin{align}
        \label{eqn:non-uniform-complexity-prefix}
        m(L_i) := \max_{\mcC'}\{|\cap_{L \in \mcC'}L|: \mcC' \subseteq (L_1,\ldots,L_i), \mcC' \ni L_i, |\cap_{L \in \mcC'}L|<\infty\}.
    \end{align}
Notice that the difference in $d_i$ and $m(L_i)$ above %
is that the latter quantity is computed as the maximum over a smaller set (due to the additional condition that $\mcC' \ni L_i$); hence, $m(L_i) \le d_i$. The guarantee given by Theorem 6 in \cite{charikar2024exploring} is that $t(L_i)=\max(i, m(L_i)+1)$.

We also observe that the generation time guarantee of \cite{charikar2024exploring} for the language in the \textit{first} position in $\mcC$ is \textit{always} 1, which is the smallest achievable generation time. %
So, for any other algorithm, if the generation time for some language $L$ in the collection is larger than 1, we can place $L$ at the start of $\mcC$, and run the algorithm of \cite{charikar2024exploring} to get an improved generation time of 1 for $L$. 
Thus, no algorithm can simultaneously achieve the smallest possible generation time for \textit{every} language $L$, unless its generation time for every language is 1.
\subsection{Other Related Work}
\label{sec:related-work}

The work of \cite{kleinberg2024language} also initiated the study of language generation in the limit with a \textit{uniform} guarantee, albeit for finite collections. In uniform generation, the generation time of an algorithm must be a global quantity for the collection, irrespective of the language being enumerated by the adversary. Recall that in non-uniform generation, the generation time may depend on the language being enumerated (but in neither setting may it depend on the enumeration order). \cite{li2024generation} formalized uniform generation more generally for collections that are not necessarily finite, and obtained a quantitatively tight characterization of the uniform generation time for a collection in terms of its closure dimension. Both these works also study language generation in the limit with a prompt. More recently, \cite{hanneke2025union} and \cite{bai2025language} established results on the behavior of non-uniform generation under unions of language collections. \cite{bai2025language} also study several other variants of generation, including generation with feedback, which was introduced by \cite{charikar2024exploring}.

Besides uniform/non-uniform generation, there has also been interest in understanding the \textit{breadth} of the target language covered by algorithms that generate in the limit. In particular, \cite{charikar2024exploring}, \cite{kalavasis2024limits} and \cite{kalavasis2024characterizations} relate definitions of generating with breadth to language identification, and establish strong negative results for the same. Very recently, \cite{kleinberg2025density} proposed more fine-grained \textit{density} measures for breadth. Their work proposes new algorithms for generation with a view to maximize these density measures, and has remarkable positive results.

\section{Pareto-optimal Non-uniform Generation}
\label{sec:Pareto-optimal-non-uniform}

We begin with a motivating example of a simple collection $\mcC$ for which the guarantees for both the algorithms of \cite{li2024generation} and \cite{charikar2024exploring} get Pareto-dominated by the latter algorithm %
when run on a suitable \textit{reordering} of $\mcC$.

\begin{restatable}{example}{exampleparetosuboptimal}
    \label{example:lrt-cp-suboptimal}
    Consider the collection $\mcC=(L_1,L_2,L_3,\ldots)$, where
    \begin{align*}
        &L_1 = \{1,\dots,99,100\} \cup \{-p_1, -p_1^2, -p_1^2, -p_1^3, \ldots\} \\
        &L_2 = \{1,\ldots,99,100\} \cup \{101,\ldots,199,200\} \cup \{201,\ldots,299,300\} \cup \{-p_2, -p_2^2, -p_2^2, -p_2^3, \ldots\} \\
        &L_3 = \{101,\ldots,199,200\} \cup \{-p_3, -p_3^2, -p_3^2, -p_3^3, \ldots\} \\
        &L_4 = \{201,\ldots,299,300\} \cup \{-p_4, -p_4^2, -p_4^2, -p_4^3, \ldots\} \\
        &L_i = \{-p_i, -p_i^2, -p_i^2, -p_i^3, \ldots\}, \quad \forall i \ge 5.
    \end{align*}
    Here, $p_n$ refers to the $n^\text{th}$ prime number. 
\end{restatable}

Denoting the closure dimension of the prefix $(L_1,\dots,L_i)$ by $d_i$ as in \Cref{appsec:example}, we observe that $d_1=0$, and $d_i =100$ for $i \ge 2$ (and hence $\lim_{i \to \infty}d_i=100$). The guarantee of \cite{li2024generation} then gives
\begin{align*}
    &t(L_i) = 101 \quad \text{for }1 \le i \le 100, \quad
    \text{ and } \quad t(L_i) = i, \; \forall i \ge 101.
\end{align*}
Similarly, we observe that $m(L_1)=0$, $m(L_2)=100$, $m(L_3)=100$, $m(L_4)=100$, and $m(L_i)=0$ for $i \ge 5$. The guarantee of \cite{charikar2024exploring} then gives
\begin{align*}
    &t(L_1) = 1, t(L_2)=101, t(L_3)=101, t(L_4)=101, \quad \text{ and } \quad t(L_i) = i, \; \forall i \ge 5.
\end{align*}
But now, consider reordering the languages in $\mcC$, so that $L_2$ is \textit{after} $L_3$ and $L_4$. That is, set $\mcC = (L_1,L_3,L_4,L_2,L_5,L_6,\ldots)$, and let us now consider the guarantee of \cite{charikar2024exploring} for this reordered collection. We obtain that $m(L_1)=0$, $m(L_3)=0$, $m(L_4)=0$, $m(L_2)=100$, and $m(L_i)=0$ for $i \ge 5$. So, we obtain the following generation times:
\begin{align*}
    &t(L_1) = 1, t(L_2)=101, t(L_3)=1, t(L_4)=1, \quad \text{ and } \quad t(L_i) = i, \; \forall i \ge 5.
\end{align*}
Observe that this sequence of generation times is \textit{strictly better} than both the above guarantees; in particular, the generation time for $L_3$ and $L_4$ is strictly smaller, while the generation time for any other language is no larger. The example above can be easily modified to an instance where $\lim_{i \to \infty}d_i = \infty$, or to make the suboptimality appear more extreme.

The structure in the example above %
motivates the following procedure for constructing a Pareto-optimal sequence of generation times.

\subsection{Pareto-optimal Sequence of Generation Times}
\label{subsec:Pareto-optimal-sequence}

We specify Procedure \ref{proc:non-uniform} which computes a non-uniform complexity $m^\star(L_i)$ for every $i \geq 1$.
In the $i^\text{th}$ iteration of this procedure, we determine $m^\star(L_i)$, and also a subcollection $\mcC(L_i)$, which is a ``witness'' for $m^\star(L_i)$. %
In the process, we maintain an ordering $\mcC'_i$ of the first $i$ languages in $\mcC$ (which may be different from the way these languages are ordered in $\mcC$). This ordering $\mcC'_i$ is maintained using insertion sort: in the $i^\text{th}$ iteration, we first append the $i^\text{th}$ language $L_i$ in $\mcC$ to the ordering $\mcC'_{i-1}$ of the first $i-1$ languages that was determined in the previous iterations. Then, we advance the position of $L_i$ in the ordering one position at a time until a certain local condition is satisfied.

\renewcommand{\figurename}{Procedure}
\begin{figure}[t]
    \begin{framed}
    \centering \textbf{Procedure for Computing Non-uniform Complexities $m^\star(L_i)$}
    \begin{enumerate}[leftmargin=0.2cm]
        \item Initialize $\mcC'_0=()$.
        \item For $i=1,2,3,\ldots$
        \begin{enumerate}[leftmargin=0.5cm]
            \item Initialize $\mcC'_i$ by appending $L_i$ at the end of $\mcC'_{i-1}$, and denote the resulting sequence $\mcC'_i$ as $(L'_1,\ldots,L'_i)$. Note that this is simply a permutation of $(L_1,\ldots,L_{i})$. %
            We want to run one iteration of ``insertion sort" so that $L_i$ is at its correct position in $\mcC'_{i}$. Initialize $j=i$.
            \item While $\true$ do:
            \begin{enumerate}[leftmargin=0.5cm]
                \item  Let $\mcC_{\chk}$ be the subcollection of $(L'_1,\ldots,L'_{j})$ that has the largest finite intersection among subcollections that contain $L'_j$. %
                If there is no such subcollection, set $\mcC_{\chk}=()$, %
                $m_{\chk}=0$. Else, set $m_{\chk} = |\cap_{L \in \mcC_{\chk}}L|$. %
                \item If $j \le 1$ or $m_{\chk} > m^\star(L'_{j-1})$, break. %
                \item Swap $L'_{j}$ and $L'_{j-1}$ in $\mcC'_i$.
                \item %
                $j \leftarrow j-1$
            \end{enumerate}
            \item Set $\mcC(L_i)=\mcC_{\chk}$, $m^\star(L_i)=m_{\chk}$.
        \end{enumerate}
    \end{enumerate}
    \end{framed}
    \caption{Insertion Sort for Non-uniform Generation}
    \label{proc:non-uniform}
\end{figure}

Some remarks about the procedure are in order. We note that $L'_j$ is always equal to $L_i$ at the beginning of every iteration of the while loop. For every $i$, $\mcC(L_i)$ and $m^\star(L_i)$ are assigned once and for all at the end of the $i^\text{th}$ iteration of the for loop; $\mcC(L_i)$ is set to some subcollection of $(L_1,\ldots,L_i)$ that contains $L_i$ and has finite intersection (if it is non-empty), and $m^\star(L_i)=|\cap_{L \in \mcC(L_i)}L|$. %
Furthermore, at the end of the $i^\text{th}$ iteration of the procedure, insertion sort guarantees that the sequence $m^\star(L'_1), m^\star(L'_2),\ldots,m^\star(L'_i)$ is non-decreasing. However, the sequence $m^\star(L_1), m^\star(L_2),m^\star(L_3),\ldots$ is \textit{not} necessarily non-decreasing. We also observe the following: %

\begin{observation}[Languages in witness have smaller complexity]
    \label{observation:S-contains-smaller-complexity}
    For every $i \ge 1$, when $\mcC(L_i)$ and $m^\star(L_i)$ are determined in Step(c) above, $\mcC(L_i)$ is either empty, or contains at least one $L_j \neq L_i$ for $j < i$, and furthermore, every $L_j \neq L_i$ in $\mcC(L_i)$ satisfies that $m^\star(L_j) < m^\star(L_i)$.
\end{observation}

\begin{proof}
    Suppose $\mcC(L_i)$ is not empty when it is determined. Note that this must mean that $j > 1$ when the while loop was exited. Note also that $\mcC(L_i)$ can't only contain copies of $L_i$,  %
    since in this case, the intersection of languages in $\mcC(L_i)$ (simply $L_i$) would be infinite (by assumption that all languages are infinite). Observe then that by virtue of the while loop exiting, it must have been the case that $m_{\chk} > m^\star(L'_{j-1})$. But recall that $m^\star(L'_{1}),\ldots,m^\star(L'_{j-1})$ is a non-decreasing sequence. The claim then follows since any language $L_j$ in $\mcC(L_i)$ (which, recall, is assigned to be $\mcC_{\chk}$) that is not equal to $L_i$ is contained in $(L'_1,\ldots,L'_{j-1})$, and the fact that we assign $m^\star(L_i)=m_{\chk}$.
\end{proof}

This observation ensures that $m^\star(\cdot)$ forms a Pareto-optimal sequence of generation times.

\begin{claim}[$m^\star(\cdot)$ forms a Pareto-optimal sequence]
    \label{claim:non-uniform-generation-lb}
    Any algorithm $\mcG$ that satisfies $t_{\mcG}(L_i) < m^\star(L_i)+1$ %
    for some $L_i$ must have $t_{\mcG}(L_j) > m^\star(L_j)+1$ for some $j < i$.
\end{claim}
\begin{proof}
    Note that the language $L_i$ must be such that $m^\star(L_i) > 0$, since $t_{\mcG}(\cdot)$ is always at least 1.%
    But this means that $\mcC(L_i) \neq ()$. Then, from \Cref{observation:S-contains-smaller-complexity}, it must be the case that at least some $L_{j} \in \mcC(L_i)$ is not equal to $L_i$, where $j < i$. Suppose an adversary enumerates all the $m^\star(L_i)$ strings in $I=\cap_{L \in \mcC(L_i)}L$ in the first $m^\star(L_i)$ time steps. Let $z$ be the string generated by the algorithm $\mcG$ at $t=m^\star(L_i)$. 
    By the guarantee of $\mcG$ that $t_\mcG(L_i) < m^\star(L_i)+1$, it must be the case that $z \in L_i \setminus I$ (otherwise, the adversary can continue to enumerate the rest of $L_i$ beyond this point, and we would have obtained an enumeration of $L_i$ for which $\mcG$ violates its guarantee for $L_i$ at $t=m^\star(L_i)$). But since we have enumerated every string in the intersection of languages in $\mcC(L_i)$, it must be the case that $z \notin L_{j}$ for some $L_{j} \in \mcC(L_i)$ where $L_{j} \neq L_i$ and $j < i$. But note also that from \Cref{observation:S-contains-smaller-complexity}, $m^\star(L_{j}) < m^\star(L_i)$. That is, if the adversary were to continue enumerating $L_{j}$ beyond $t$, and $t_{\mcG}(L_{j})$ were at most $m^\star(L_{j})+1$, which is at most $m^\star(L_i)$, then the string $z$ generated by $\mcG$ at this time should have belonged to $L_{j}$, which we know is not the case. Thus, $t_\mcG(L_j) > m^\star(L_j)+1$.
\end{proof}

\subsection{(Almost) Pareto-optimal Non-uniform Generation Algorithm}
\label{subsec:Pareto-optimal-algorithm}

While Procedure \ref{proc:non-uniform} specifies a sequence of Pareto-optimal generation times, it also motivates a corresponding algorithm that \textit{almost} attains this sequence. To begin, let us think about what happens when we swap $L'_j$ and $L'_{j-1}$ within an iteration of insertion sort. Note that after moving $L'_{j}$ to position $j-1$, the collection $\mcC_{\chk}$ is obtained as the max over subcollections that are \textit{strictly} contained among the subcollections under consideration before moving; hence, it is clear that $m_{\chk}$ can only decrease. More importantly, we also have the following nice property, which turns out to be crucial for the correctness of the algorithm we state below: if one were to compute $m_{\chk}$ for $L'_{j-1}$ after it has been moved forward to position $j$, it remains equal to $m^\star(L'_{j-1})$, \textit{even if} $m_{\chk}$ is now a maximum over a strictly larger space---this is a consequence of the while loop not exiting.

\begin{claim}[Arg max Maintained]
    \label{claim:argmax-maintained}
    For every $i \ge 1$, and for every $t \ge i$, suppose $k$ is the index of language $L_i$ in $\mcC'_t=(L'_1,\ldots,L'_t)$ at the end of the $t^\text{th}$ iteration of the for loop. Then,
    \begin{align}
        \label{eqn:argmax-maintained}
        \mcC(L_i) \in \argmax_{\mcC_{\chk}}\{|\cap_{L \in \mcC_\chk}L|:\mcC_{\chk} \subseteq (L'_1,\ldots,L'_k), \mcC_{\chk} \ni L'_k, \left|\cap_{L \in \mcC_{\chk}}L\right|<\infty\}.
    \end{align}
\end{claim}
\begin{proof}
    For $t=i$, the claim is true by definition of $\mcC(L_i)$. Now suppose as our inductive hypothesis that the claim is true for some $t \ge i$. We will show that it is true for $t+1$. Suppose $k$ is the index of $L_i$ at the end of the $t^\text{th}$ iteration of the for loop. If $k$ continues to be the index of $L_i$ at the end of the $(t+1)^\text{th}$ iteration as well (meaning that the language $L_{t+1}$ was placed at some index larger than $k$ in $\mcC'_{t+1}$), the claim continues to be true simply by the inductive hypothesis. Otherwise, the language $L_{t+1}$ was placed at an index in $\{1,\ldots, k\}$, and so $L_i$ was moved to index $k+1$ in $\mcC'_{t+1}$. In this case, consider the precise iteration of the while loop when $j$ was equal to $k+1$ (so that $L'_{k+1}$ was $L_{t+1}$), and $m_{\chk}$ was found to be at most $m^\star(L'_k)$ in Step ii. At this point, Step iii tells us to swap $L'_k$ and $L'_{k+1}$ (this will make $L'_{k+1}$ equal to $L_i$, and $L'_k$ equal to $L_{t+1}$). Thereafter, note that the (unordered) subcollection of languages before $L'_{k+1}$ remains unchanged until the while loop breaks, and the present iteration of the for loop completes. So, it suffices to argue that, \textit{before} swapping,
    \begin{align*}
        \mcC(L_i) \in \argmax_{\mcC_{\chk}}\{|\cap_{L \in \mcC_\chk}L|:\mcC_{\chk} \subseteq (L'_1,\ldots,L'_k, L'_{k+1}), \mcC_{\chk} \ni L'_k, \left|\cap_{L \in \mcC_{\chk}}L\right|<\infty\}.
    \end{align*}
    Recall that our inductive hypothesis already guarantees that $\mcC(L_i)$ belongs to the arg max over all feasible $\mcC_{\chk}$, where the condition on $\mcC_{\chk}$ only considers subcollections of $(L'_1,\ldots,L'_{k})$. That is, the only additional language now being considered in the condition is $L'_{k+1}$. So, if the max were to increase from what it was when $L'_{k+1}$ was not being considered, it must be the case that the $\mcC_{\chk}$ which is the new $\argmax$ \textit{contains} $L'_{k+1}$. 
    But this means that before swapping, there exists a subcollection of $(L'_1,\ldots,L'_{k}, L'_{k+1})$ that contains $L'_{k+1}$, and has finite intersection of size strictly larger than $|\cap_{L \in \mcC(L_i)}L|$. This contradicts the fact that $m_{\chk}$ was found to be at most $m^\star(L'_k)$ (since $L'_k=L_i$, and $|\cap_{L \in \mcC(L_i)}L|=m^\star(L_i)$). Thus, it must be the case that $\mcC(L_i)$ continues to be in the $\argmax$ even after swapping. This completes the proof of the inductive step.
\end{proof}

\Cref{claim:argmax-maintained} motivates the following algorithm for Pareto-optimal non-uniform generation:

\paragraph{Algorithm.} %
Fix some non-decreasing function $f: \N \to \N$ that satisfies $\lim_{t \to \infty}f(t)=\infty$. At time step $t$, the algorithm follows $f(t)$ iterations %
of Procedure \ref{proc:non-uniform} for the purpose of constructing the ordered collection $\mcC'_{f(t)}=(L'_1,\ldots,L'_{f(t)})$ (a permutation of the first $f(t)$ languages in $\mcC$). %
It then initializes $I_t=()$, and traverses $L'_1,L'_2,\ldots,L'_{f(t)}$ in this order. Whenever it encounters a language $L'_j$ that satisfies $S_t \subseteq L'_j$, it checks if $|\cap_{L \in I_t \cup \{L'_j\}}L|=\infty$. If so, it appends $L'_j$ to $I_t$; %
otherwise, it moves on, leaving $I_t$ unaffected. Observe that at the end of the traversal, $I_t$ is maintained to be a collection of languages that has infinite intersection (if it is non-empty). So, at the end of its traversal, if $I_t$ is empty, the algorithm generates an arbitrary new string from the universe $U$, and if $I_t$ is non-empty, the algorithm generates a new string from $\cap_{L \in I_t}L$. 

We note that the algorithm of \cite{charikar2024exploring} is essentially the algorithm above, with the choice of $f(t)=t$, and $\mcC'_{f(t)}=(L_1,\ldots,L_{f(t)})$ (i.e., the default ordering in $\mcC$). The crucial change is the carefully maintained ordering $(L'_1,\ldots,L'_{f(t)})$. We can use \Cref{claim:argmax-maintained} to derive the following characterization of non-uniform generation times for the algorithm:

\begin{theorem}[Non-uniform Generation Times]
    \label{theorem:non-uniform-generation-ub}
    The algorithm described above non-uniformly generates from the collection $\mcC=(L_1,L_2,\ldots)$, with $t^\star(L_i) = \max(g(i), m^\star(L_i)+1)$, where $g(i)$ is the smallest number $j$ for which $f(j) \ge i$.
\end{theorem}
\begin{proof}
    Suppose that the target language being enumerated by the adversary is $L_i$. Consider any time $t$ for which $|S_t| \ge t^\star(L_i)$; this implies $t \ge t^\star(L_i) \ge g(i)$. Since $f$ is non-decreasing, the language $L_i$ is under consideration by the algorithm at time $t$. So, let $k$ be the index of the language $L_i$ in the collection $\mcC'_{f(t)}=(L'_1,\ldots,L'_{f(t)})$ constructed by the algorithm at time step $t$. We only need to argue that when $L'_k$ is encountered by the algorithm as it traverses $\mcC'_{f(t)}$, it finds that $|\cap_{L \in I_t \cup \{L'_k\}}L|=\infty$ (note that $S_t \subseteq L'_k$ holds, since $L'_k$ is the target language being enumerated). Assume for the sake of contradiction that $|\cap_{L \in I_t \cup \{L'_k\}}L| < \infty$. In this case, observe that $I_t \cup \{L'_k\}$ is a subcollection of $(L'_1,\ldots,L'_k)$ which contains $L'_k$ and has finite intersection. Furthermore, the size of this intersection is at least $m^\star(L_i)+1$ since every language in the subcollection contains $S_t$ by definition of the algorithm. This would imply that the maximum in the RHS of \eqref{eqn:argmax-maintained} in \Cref{claim:argmax-maintained} is at least $m^\star(L_i)+1$, which contradicts $\mcC(L_i)$ realizing the maximum, since $|\cap_{L \in \mcC(L_i)}L| = m^\star(L_i)$. Thus, $L'_k=L_i$ gets added to $I_t$ when it is encountered by the algorithm; since the algorithm generates a new string from $\cap_{L \in I_t}L$ which is ensured to be an infinite set, we are done.
\end{proof}

\begin{remark}
    \label{remark:better-than-cp}
    The guarantee above is strictly no worse than the non-uniform generation guarantee of \cite{charikar2024exploring} (see \Cref{appsec:example}), if we substitute $f(t)=t$. %
\end{remark}

\begin{remark}
    \label{remark:almost-Pareto-optimal}
    Observe that by choosing $f$ to be arbitrarily fast-growing, we can ensure that $t^\star(L_i)=m^\star(L_i) + 1$ for all $i \le n$, for an arbitrarily large $n$. In this sense, and as also stated in \Cref{thm:informal-almost-Pareto-optimal-non-uniform-language-generation}, our algorithm above almost attains Pareto-optimality, since it can match the generation times of a Pareto-optimal sequence (\Cref{claim:non-uniform-generation-lb}) in an arbitrarily large prefix .
\end{remark}

\subsection{Sufficient Condition for Exact Pareto-optimality}
\label{subsec:possibility-of-Pareto-optimality}

We identify a sufficient condition for a collection that guarantees \textit{exact} Pareto-optimality for our algorithm. %
Here,
for an appropriately chosen function $f$,
our algorithm 
matches the Pareto-optimal sequence of generation times identified by the procedure above (\Cref{claim:non-uniform-generation-lb}) at \textit{every} language.

\begin{theorem}[Sufficient Condition for Exact Pareto-optimality]
    \label{thm:Pareto-optimality-guarantee}
    Let $\mcC=(L_1,L_2,\ldots)$ be a collection that satisfies the following property: for every $t \in \N$, $|\{j: m^\star(L_j)+1\le t\}| < \infty$, where the $m^\star(\cdot)$ values are those computed in Procedure \ref{proc:non-uniform}. Then, there exists an algorithm that achieves a Pareto-optimal sequence of non-uniform generation times for $\mcC$.
\end{theorem}
\begin{proof}
    Consider $f:\N \to \N$ defined as $f(t)=\max\{j:m^\star(L_j)+1 \le t\}$. Observe that $f$ is a non-decreasing function, and that $\lim_{t \to \infty}f(t)=\infty$. Furthermore, observe also that for any $i$, $f(m^\star(L_i)+1) \ge i$. Therefore, if $g(i)$ is the smallest number $j$ for which $f(j) \ge i$, then $g(i) \le m^\star(L_i)+1$. So, if we choose to run the algorithm given in \Cref{subsec:Pareto-optimal-algorithm} with this function $f$, the guarantee of \Cref{theorem:non-uniform-generation-ub} implies $t^\star(L_i)=m^\star(L_i)+1$ for every $L_i$. From \Cref{claim:non-uniform-generation-lb}, we conclude that the algorithm is Pareto-optimal.
\end{proof}

A direct corollary of \Cref{remark:almost-Pareto-optimal} and \Cref{thm:Pareto-optimality-guarantee} is that exact Pareto-optimal non-uniform generation can be achieved for every \textit{finite} collection of languages.

\subsection{Impossibility of Pareto-optimality}

It is natural to wonder if the sufficient condition stated above is also necessary: this would exactly characterize the collections for which achieving Pareto-optimality is possible; moreover, our algorithm above would achieve such Pareto-optimality. Towards this, we exhibit two collections, for which we show that there \textit{cannot} exist an algorithm that achieves a Pareto-optimal sequence of generation times. One of these collections can be non-uniformly, but not uniformly generated, while the other can even be \textit{uniformly} generated. Both these collections \textit{fail} the sufficient condition stated in \Cref{thm:Pareto-optimality-guarantee}. We leave open the complete characterization of Pareto-optimality for future work.

\begin{restatable}[Non-uniformly Generatable, No Pareto-optimality]{proposition}{Paretoimpossibilitynonuniform}
    \label{claim:Pareto-impossibility-non-uniform}
    No algorithm can achieve Pareto-optimal non-uniform generation times for the collection $\mcC=\{\Z \setminus \{i\}\}_{i \in \Z}$.
\end{restatable}
\begin{proof}
    We note first that the collection in question cannot be uniformly generated because its closure dimension (see \Cref{appsec:example}) is unbounded. To see this, observe that for any finite $S \subseteq \Z$, the intersection of languages in $\mcC$ that contain $S$ is precisely $S$. However, this collection is non-uniformly generatable, since it is countable.

    Let $L_{e_1}, L_{e_2}, L_{e_3},\ldots$ be an enumeration of all the languages in $\mcC$ (with no repetition), where $L_{e_i} = \Z \setminus \{e_i\}$ for $e_i \in \Z$. Let $t_1,t_2,t_3,\ldots$ be a sequence of candidate generation times for these languages, where every $t_i$ is finite and at least 1. We call the sequence \textit{admissible} if it satisfies the following property: for every $t \ge 1$, the size of the set $\{i:t_i \ge t\}$ is infinite. We claim that there exists an algorithm that achieves the generation times given by the sequence $t_1,t_2,\ldots$ if and only if it is admissible. 
    
    For one direction, suppose that the sequence is not admissible, meaning that there exists some $t$ for which $|\{i:t_i \ge t\}| < \infty$. Suppose that an algorithm achieves the generation times given in the sequence. Then, observe that the said algorithm \textit{uniformly} generates from the collection, with a uniform time bound of $\max\{t_i: t_i \ge t\}$. Since we argued above that the collection cannot be uniformly generated, such an algorithm cannot exist.

    For the other direction, suppose that the sequence is admissible. Consider the algorithm that operates as follows: at time step $t$, it obtains the subcollection $\mcC_t = \{L_{e_i}: t_i \le t\}$. If $\mcC_t$ is non-empty, and $\cap_{L \in \mcC_t}L$ is also non-empty, the algorithm generates a new string from $\cap_{L \in \mcC_t}L$ (or repeats a string if there are no new strings available in this set); otherwise, it generates an arbitrary string. Note that this algorithm is oblivious to its input.

    We claim that for every language $L_{e_i} \in \mcC$, this algorithm achieves the desired generation time of $t_i$. To see this, it suffices to argue that for every $t$ that satisfies $|S_t| \ge t_i$, $\mcC_t$ contains $L_{e_i}$, and that $\cap_{L \in \mcC_t}L$  is infinite. Note first that any $t$ satisfying $|S_t| \ge t_i$ immediately satisfies $t \ge t_i$ (since in $t$ times steps, the input can have at most $t$ distinct strings). By definition of the algorithm then, $\mcC_t$ contains $L_{e_i}$; in particular, it is non-empty.

    Now, since the sequence is admissible, the set $\{j:t_j \ge t+1\}$ is infinite. In particular, there exist infinitely many languages $L_{e_j}$, such that each $L_{e_j}$ is excluded from $\mcC_t$. But observe that for any $L_{e_j}$ that is excluded from $\mcC_t$, \textit{every} language in $\mcC_t$ contains $e_j$. Since there are infinitely many such $L_{e_j}$, we conclude that $\cap_{L \in \mcC_t}L$ is infinite.

    To finish the proof, assume for the sake of contradiction that there exists an algorithm that achieves non-uniform generation times $t_1,t_2,\ldots$ for $L_{e_1}, L_{e_2},\ldots$ which are Pareto-optimal. By our reasoning above, the sequence $t_1,t_2,\ldots$ must be admissible. But now, consider any $t_i > 1$ in the sequence (such a $t_i$ must exist by admissibility), and consider updating $t_i=1$ (other generation times in the sequence are left unaffected). This updated sequence strictly Pareto-dominates the previous sequence. Furthermore, the updated sequence remains admissible, since we only changed one number. This means that there exists an algorithm that achieves the generation times in the updated sequence. Consequently, the algorithm we started with could not have been Pareto-optimal.
\end{proof}

\begin{restatable}[Uniformly Generatable, No Pareto-optimality]{proposition}{Paretoimpossibilityuniform}
    \label{claim:Pareto-impossibility-uniform}
    No algorithm can achieve Pareto-optimal non-uniform generation times for the collection $\mcC$ comprising of languages $\{1,2,\ldots,100\} \cup \Z_{\ge i}$ for all $i \in \Z$, where $\Z_{\ge i}=\{i,i+1,i+2,\ldots\}$.
\end{restatable}

We defer the proof of this result to \Cref{sec:Pareto-optimalilty-impossibility-uniform}. We conclude this section by arguing that both the collections from \Cref{claim:Pareto-impossibility-non-uniform} and \Cref{claim:Pareto-impossibility-uniform} do not satisfy the sufficient condition of Pareto-optimality in \Cref{thm:Pareto-optimality-guarantee}. To see this, for either of these collections $\mcC$, note that the intersection of languages in any finite subcollection is infinite; since the $m^\star(\cdot)$ values in the Pareto-optimal sequence computed in \Cref{subsec:Pareto-optimal-sequence} correspond to finite intersections of finite subcollections, we have that $m^\star(L_j)=0$ for every $L_j \in \mcC$ (in fact, this holds true irrespective of the original ordering of $\mcC$ that we start with to compute the $m^\star(\cdot)$ values). Thus, the sufficient condition does not hold for $t=1$, since $|\{j:m^\star(L_j)+1 \le 1\}|=\infty$.

We now proceed to study Pareto-optimality for two recent stylized models of generation: noisy and representative generation.
\section{Pareto-optimal Noisy Non-uniform Generation}
\label{sec:Pareto-optimal-noisy}

In noisy generation \citep{raman2025generation}, we consider \textit{noisy enumerations} of the target language. For a language $L$, an enumeration of $L$ at (finite) noise level $n \ge 0$ is a sequence $x_1,x_2,x_3,\ldots$, which satisfies: 
(1) For every $x \in L$, there exists $t < \infty$ for which $x_t=x$, and 
(2) $\sum_{t=1}^\infty \Ind[x_t \notin L] \le n$.

Note that the noiseless setting considered in the previous section is only concerned with enumerations at noise level $n=0$. For a set $T$, language $L$, and integer $a \ge 0$, we will say that \textbf{$L$ $a$-contains $T$ (or $T$ is $a$-contained in $L$)} if $\sum_{x \in T} \Ind[x \notin L] \le a$.

\begin{definition}[Noisy Non-uniform Generation \citep*{raman2025generation}]
    \label{def:noisy-non-uniform-generation}
    An algorithm $\mcG$ noisily non-uniformly generates in the limit from a collection $\mcC$, if for every $L \in \mcC$ and every finite noise level $n \ge 0$, there exists $t_{n}(L)$ such that for every enumeration of $L$ at noise level $n$ presented to $\mcG$, the string $z_t$ generated by $\mcG$ at time step $t$ belongs to $L \setminus S_t$ for all $t$ satisfying $|S_t| \ge t_n(L)$.
\end{definition}

\cite{raman2025generation} derive comprehensive results characterizing noisy (non-uniform as well as uniform) generation; however, examples like \Cref{example:lrt-cp-suboptimal} render their algorithm to be Pareto-suboptimal as well; we will now work towards a Pareto-optimal noisy non-uniform generation algorithm.

\setcounter{figure}{0} 
\renewcommand{\figurename}{Figure}
\begin{figure}[t]
    \centering
    \includegraphics[scale=0.32]{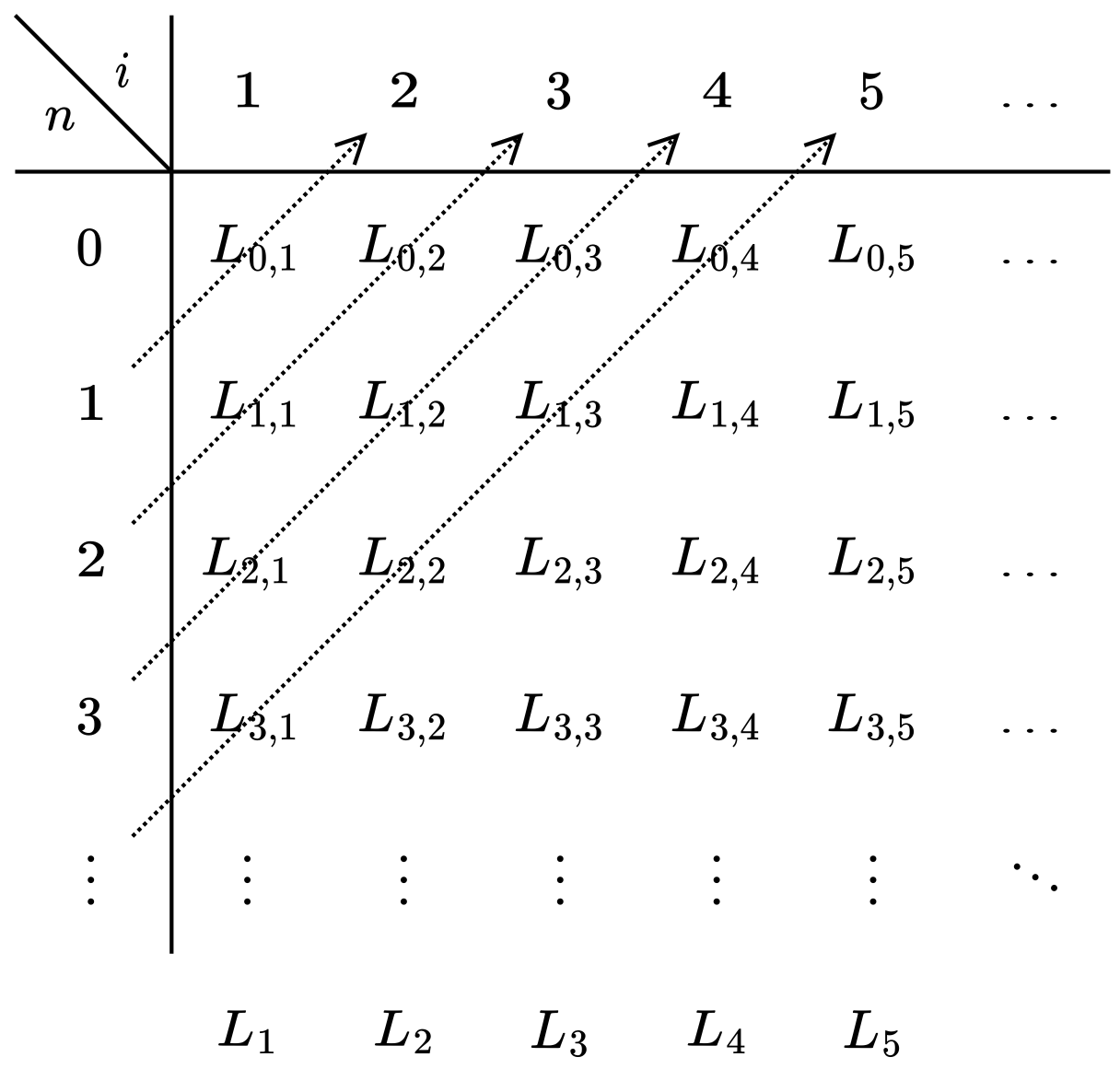}
        \caption{Diagonal traversal over languages arranged in a grid. The $i^{th}$ column entirely consists of $L_i$. When we arrive at a copy of $L_i$ in the $n^{\text{th}}$ row (i.e., at $L_{n,i}$), we compute $m^\star_n(L_i)$.}
        \label{fig:2d-grid}
\end{figure}

Consider arranging the languages in $\mcC$ in a two-dimensional grid, rows indexed by the noise level $n \ge 0$, and columns by the index $i \ge 1$ of the language in $\mcC$. For each row $n$, the element in the $i^\text{th}$ column, denoted $L_{n, i}$, is simply a copy of $L_i$. 
Rows correspond to %
the different noise levels %
for each language. We perform a diagonal traversal (\Cref{fig:2d-grid}) of the two-dimensional grid: let $\diag(\mcC_{n,i})$ denote the sequence of languages %
in a diagonal traversal starting at $L_{0,1}$ and reaching $L_{n,i}$. That is, $\diag(\mcC_{n,i})=(L_{0,1}, L_{1,1}, L_{0,2}, L_{2,1}, L_{1,2}, L_{0,3}, L_{3,1},L_{2,2},L_{1,3},L_{0,4},\ldots, L_{n, i})$.

For every $i \ge 1$ and $n \ge 0$, we will compute a noisy non-uniform complexity $m_n^\star(L_i)$ in Procedure \ref{proc:noisy}. Our approach will be similar to the noise-free setting from above, where we will perform insertion sort on our diagonal traversal. That is, as we traverse, we will keep track of an ordering of $\diag(\mcC_{n,i})$, which we will denote as $\mcC'_{l}$. In each iteration of the traversal, we will determine $m^\star_n(L_i)$, and also a subcollection $\mcC(L_i)$ and set $T(L_i)$ which together act as ``witnesses'' for $m^\star_n(L_i)$. Our algorithm for noisy non-uniform generation is also based on this diagonal traversal.

\setcounter{figure}{1} 
\renewcommand{\figurename}{Procedure}
\begin{figure}[t]
    \begin{framed}
    \centering \textbf{Procedure for Computing Noisy Non-uniform Complexities $m^\star_n(L_i)$}
    \begin{enumerate}[leftmargin=0.2cm]
        \item Initialize $\mcC'_{0}=()$, $l=1$.
        \item For $n'=0,1,2,\ldots$
        \begin{enumerate}[leftmargin=0.5cm]
            \item For $h=0,\ldots,n'$
                \begin{enumerate}[leftmargin=0.5cm]
                    \item Set noise level $n=n'-h$, language index $i=h+1$.
                    \item Initialize $\mcC'_l$ by appending $L_{n, i}$ at the end of $\mcC'_{l-1}$, and denote the resulting sequence $\mcC'_l$ as $(L'_{1},\ldots,L'_{l})$. This is simply a permutation of $\diag(\mcC_{n,i})$. %
                    We want to run one iteration of ``insertion sort" so that $L_{n, i}$ is at its correct position in $\mcC'_{l}$. Initialize $j=l$.
                    \item While $\true$ do:
                        \begin{enumerate}[leftmargin=0.5cm]
                            \item  Let $T$ be the largest \textit{finite} set such that there exists a subcollection $\mcC_{\chk} \subseteq (L'_1,\ldots,L'_j)$ satisfying the following conditions:
                            \begin{itemize}[leftmargin=0.5cm]
                                \item $\mcC_{\chk} \ni L'_j$.
                                \item For every $L' \in \mcC_{\chk}$, which corresponds to some $L_{a,b}$, %
                                $T$ is $a$-contained in $L_{b}$.
                                \item The intersection of languages in $\mcC_{\chk}$ is finite.
                            \end{itemize}
                            Let $m_{\chk}=|T|$.
                            \item Suppose $L'_{j-1}$ corresponds to $L_{a,b}$. If $j \le 1$ or $m_{\chk} > m^\star_a(L_b)$, break.
                            \item Swap $L'_{j}$ and $L'_{j-1}$ in $\mcC'_l$.
                            \item %
                            $j \gets j - 1.$
                        \end{enumerate}
                    \item Set $T(L_{n,i})=T$, $\mcC(L_{n,i})=\mcC_{\chk}, m^\star_{n}(L_{i})=m_{\chk}$.
                    \item %
                    $l \gets l+1$.
                \end{enumerate}
        \end{enumerate}
    \end{enumerate}
    \end{framed}
    \caption{Insertion Sort for Noisy Non-uniform Generation}
    \label{proc:noisy}
\end{figure}

\paragraph{Algorithm.} Fix a non-decreasing function $f:\N \to \N$ which satisfies $\lim_{t \to \infty}f(t)=\infty$. At time step $t$, the algorithm follows the procedure of diagonal traversal described through $l=f(t)$ for the purpose of constructing $\mcC'_{f(t)}=(L'_1,\ldots,L'_{f(t)})$. It then initializes $I_t=()$, and traverses $L'_1,L'_2,\ldots,L'_{f(t)}$ in this order. Whenever it encounters a language $L'_j$---which corresponds to some $L_{a,b}$---, it checks if $S_t$ is $a$-contained in $L_b$, and that $|\cap_{L \in I_t \cup \{L'_j\}}L|=\infty$. If so, it appends $L'_j$ to $I_t$. Otherwise, it moves on, leaving $I_t$ unaffected. Observe that at the end of the traversal, $I_t$ is maintained to be a collection of languages that has infinite intersection (if it is non-empty). So, at the end of its traversal, if $I_t$ is empty, the algorithm generates an arbitrary new string from the universe $U$, and if $I_t$ is non-empty, the algorithm generates a new string from $\cap_{L \in I_t}L$. 

Using arguments similar to \Cref{sec:Pareto-optimal-non-uniform}, we show that the above algorithm is almost Pareto-optimal.

\begin{restatable}[Noisy Generation times]{theorem}{generationtimesnoisy}
    \label{theorem:noisy-non-uniform-gen-ub}
    The algorithm described for noisy generation noisily non-uniformly generates from the collection $\mcC=(L_1,L_2,\ldots,)$, with $t^\star_n(L_i) = \max(g(n, i), m^\star_n(L_i)+1)$, where $g(n, i)$ is the smallest number $j$ that satisfies $f(j) \ge |\diag(\mcC_{n,i})|$.
\end{restatable}

Notice again that by choosing $f$ to be arbitrarily fast-growing, we can have the $\max$ equate to $m^\star_n(L_i)+1$ for an arbitrarily large prefix. We defer the details, and also the specification of the sufficient condition which guarantees exact Pareto-optimality for the algorithm, to \Cref{appsec:Pareto-optimal-noisy}.

\section{Pareto-optimal Representative Non-uniform Generation}
\label{sec:Pareto-optimal-representative}

In representative generation \citep{peale2025representative}, we are given a finite partition\footnote{See Lemma 2 in \cite{peale2025representative} for pathologies arising with infinitely many groups.} $\mcA=\{A_g\}_{g \in [K]}$ of the universe $U$ into $K$ groups. As before, an adversary chooses and enumerates the target language $L$ in a worst-case fashion. But now, the generating algorithm $\mcG$ outputs a \textit{distribution} $\mcG_{t}$ over strings at every time step. %
Letting $S_t$ denote the set of distinct strings in $x_1,\ldots,x_t$, we want the distribution over \textit{groups} that $\mcG_{t}$ induces to match the \textit{empirical} distribution over groups seen so far. That is, let $\emp_{t}$ denote the empirical distribution over $U$ induced by $S_t$, i.e., 
\begin{align}
    \label{eqn:emp-dist-strings-def}
    \emp_{t}(x) = \frac{\Ind[x \in S_t]}{|S_t|} \quad \text{for $x \in U$}.
\end{align}
We slightly abuse notation, and use $\emp_t$ and $\emp_{S_t}$ to refer to the same object---this will be convenient when referring to the empirical distribution induced more generally by a given set. For a distribution $D$ over $U$, let $D^{\mcA}$ be the distribution it induces over groups. That is, for $g \in [K]$,
\begin{align}
    D^\mcA(g) = \Pr_{x \sim D}[x \in A_g].
\end{align}

\begin{definition}[$\alpha$-representation]
    \label{def:alpha-representation}
    For any accuracy parameter $\alpha \in [0,1]$, we say that the generator $\mcG$ is $\alpha$-representative at time step $t$ with respect to the input $S_t$ if $\|\emp_{t}^{\mcA}-\mcG_{t}^{\mcA}\|_\infty \le \alpha$.
\end{definition}

\begin{definition}[Representative Non-uniform Generation \citep*{peale2025representative})]
    \label{def:representative-non-uniform-generation}
    Fix an accuracy parameter $\alpha \in [0,1]$. A generating algorithm $\mcG$ generates non-uniformly in the limit with $\alpha$-representation from a collection $\mcC = (L_1,L_2,\ldots)$ if for every $L \in \mcC$, there exists a $t_{\alpha}(L)$ such that for any enumeration of $L$ presented to the generator, the output distribution $\mcG_{t}$ satisfies that $\Pr_{x \sim \mcG_{t}}[x \in L \setminus S_t]=1$ for every $t$ satisfying $|S_t| \ge t_\alpha(L)$, and furthermore, the output group distribution $\mcG_{t}^\mcA$ induced by the generator is $\alpha$-representative with respect to $S_t$ at every $t \ge 1$.
\end{definition}

For representative generation as well, \cite{peale2025representative} obtain numerous characterizations; our focus will be on Pareto-optimality. We require the following definition:

\begin{definition}[Scarce Groups]
    \label{def:scarce-groups}
    Let $\mcC$ be a collection of languages, $\mcA=\{A_g\}_{g \in [K]}$ be a finite partition of the universe $U$ into $K$ groups, and $T$ be a finite subset of $U$. The scarce groups in $\mcA$ with respect to $\mcC$ and $T$ correspond to the set $B := \{g \in [K]: A_g \cap \left(\cap_{L \in \mcC}L \setminus T\right)=\emptyset\}$.
    We say that $T$ suffers group scarcity (with respect to $\mcC$ and $\mcA$), if either there is a scarce group $g \in B$ that satisfies $\emp_{T}^{\mcA}(g) > \alpha$, or it is the case that $\sum_{g \in B}\emp_{T}^{\mcA}(g) > \alpha \cdot (K - |B|)$.
\end{definition}

At a high level, the set of scarce groups are precisely those groups, for which we are out of ``safe'' new strings to generate, assuming the target language is one of the languages in $\mcC$ that contains $T$. 
We will compute the optimal $\alpha$-representative non-uniform complexity $m^\star_\alpha(L_i)$ for every $L_i$ in Procedure \ref{proc:representative} using a similar insertion sort procedure as before, with an appropriately chosen comparison criterion. The procedure motivates the following algorithm, which is qualitatively quite different from the previous two algorithms, requiring a more careful distribution of the generator's probability mass amongst strings from different groups.

\begin{figure}[t]
    \begin{framed}
    \centering \textbf{Procedure for Computing Representative Non-uniform Complexities $m^\star_\alpha(L_i)$}
    \begin{enumerate}[leftmargin=0.2cm]
        \item Initialize $\mcC'_0=\emptyset$.
        \item For $i=1,2,3,\ldots$
        \begin{enumerate}[leftmargin=0.5cm]
            \item Initialize $\mcC'_i$ by appending $L_i$ at the end of $\mcC'_{i-1}$, and denote the resulting sequence $\mcC'_i$ as $(L'_1,\ldots,L'_i)$. %
            Initialize $j=i$.
            \item While $\true$ do:
            \begin{enumerate}[leftmargin=0.5cm]
                \item  Let $T$ be the largest \textit{finite} set such that there exists a subcollection $\mcC_{\chk} \subseteq (L'_1,\ldots,L'_j)$ %
                satisfying the following conditions:
                    \begin{enumerate}[leftmargin=0.5cm]
                        \item $\mcC_{\chk} \ni L'_j$.
                        \item Every $L \in \mcC_{\chk}$ satisfies that $L \supseteq T$.
                        \item $T$ suffers group scarcity with respect to $\mcC_{\chk}$ and $\mcA$.
                    \end{enumerate}
                Let $m_{\chk} = |T|$.
                \item If $j \le 1$ or $m_{\chk} > m^\star_\alpha(L'_{j-1})$, break.
                \item Swap $L'_{j}$ and $L'_{j-1}$ in $\mcC'_i$.
                \item %
                $j \gets j - 1$.
            \end{enumerate}
            \item Set $T(L_i)=T, \mcC(L_i)=\mcC_{\chk}$, $m^\star_{\alpha}(L_i)=m_{\chk}$.     
        \end{enumerate}
    \end{enumerate}
    \end{framed}
    \caption{Insertion Sort for Representative Non-uniform Generation}
    \label{proc:representative}
\end{figure}

\paragraph{Algorithm.} %
Fix a non-decreasing function $f:\N \to \N$ that satisfies $\lim_{t \to \infty}f(t)=\infty$. At time step $t$, the algorithm follows Procedure \ref{proc:representative} up until $i=f(t)$ for the purpose of constructing the ordered collection $\mcC'_{f(t)}=(L'_1,\ldots,L'_{f(t)})$. It then initializes $I_t=()$, and traverses $L'_1,L'_2,\ldots,L'_{f(t)}$ in this order. Whenever it comes across a language $L'_j$ %
that satisfies $L'_j \supseteq S_t$, it checks if $S_t$ \textit{does not} suffer group scarcity with respect to $I_t \cup \{L'_j\}$ and $\mcA$. Namely, for the set $B$ of scarce groups in $\mcA$ with respect to $I_t \cup \{L'_j\}$ and $S_t$,
it checks that \textit{every} scarce group $g \in B$ satisfies $\emp_{t}^{\mcA}(g) \le \alpha$, \textit{and} that $\sum_{g \in B}\emp_{t}^{\mcA}(g) \le \alpha \cdot (K-|B|)$. 
If $L'_j$ passes this check, it gets appended to $I_t$.

At the end of its traversal, if the algorithm finds that $I_t$ is empty, then it sets $\mcG_{t}=\emp_t$. On the other hand, if $I_t$ is non-empty, then consider the scarce groups $B$ with respect to $I_t$ and $S_t$.
Observe first that $B \subsetneq [K]$, since when the last language was inserted into $I_t$, it was found that $\sum_{g \in B}\emp_t^\mcA(g) \le \alpha \cdot (K - |B|)$; if $B=[K]$, the RHS would be 0, but the LHS is positive. 

Now, for every non-scarce group $g \in [K] \setminus B$, there exists some string $s_g \in A_g \cap \left(\cap_{L \in I_t}L \setminus S_t\right)$; let us arbitrarily pick such an $s_g$. We will compute the mass that $\emp_{t}$ assigns to the group $g$, and assign all of this mass to $s_g$. That is, $\mcG_t(s_g)=\emp_t^{\mcA}(g)$, and $\mcG_t(s')=0$ for every $s' \in A_g, s' \neq s_g$. After we do this for all of the non-scarce groups, $\mcG_t$ will have assigned a total probability mass of $\sum_{g \in [K] \setminus B}\emp_{t}^{\mcA}(g)$, and it still needs to assign a mass worth $\sum_{g \in B}\emp_{t}^{\mcA}(g)$. $\mcG_t$ will simply distribute this remaining mass evenly over all the strings $s_g$ that were chosen from the non-scarce groups $g \in [K] \setminus B$. Summarily, $\mcG_t$ has non-zero mass on at most $K-|B|$ strings; each of these strings belong to \textit{all} the languages in $I_t$, none of these strings belong to $S_t$, and none of these strings belong to a scarce group. This completes the specification of $\mcG_t$.

For the algorithm specified above, we can show the following almost Pareto-optimal guarantee with the freedom to choose $f$ as suitable; see \Cref{appsec:Pareto-optimal-representative} for details, and also for a sufficient condition that guarantees exact Pareto-optimality. The proof for $\alpha$-representation crucially uses the manner in which we spread out the generator's mass in the algorithm, together with the fact that we maintain the set of input strings $S_t$ to \textit{not} suffer group scarcity with respect to the collection $I_t$ and $\mcA$. 

\begin{restatable}[Representative Generation times]{theorem}{generationtimerepresentative}
    \label{theorem:representative-non-uniform-gen-ub}
    The algorithm described for representative generation generates non-uniformly with $\alpha$-representation from the collection $\mcC=(L_1,L_2, \ldots)$, with $t^\star_\alpha(L_i)=\max(g(i), m^\star_\alpha(L_i)+1)$, where $g(i)$ is the smallest number $j$ for which $f(j) \ge i$.
\end{restatable}

\section*{Acknowledgements}
This work was supported by Moses Charikar's and Gregory Valiant's Simons Investigator Awards.

\bibliographystyle{plainnat}
\bibliography{references}

\appendix

\section{Impossibility of Pareto-optimality for a Collection That Can Be Uniformly Generated}
\label{sec:Pareto-optimalilty-impossibility-uniform}

\Paretoimpossibilityuniform*
\begin{proof}
    We first note that the collection can be uniformly generated. To see this, observe that any finite set of size larger than 100 contains a number outside $\{1,2,\ldots,100\}$---let $i$ be the smallest such number in the set. Then, the languages that contain the set are precisely the languages $\{1,2,\ldots,100\} \cup \Z_{\ge j}$ where $j \le i$; in particular, the intersection of these languages contains all the numbers that are at least $j$, and is hence infinite. Thus, the closure dimension of the collection is at most 100. In fact, it is equal to 100, since the intersection of all the languages in the collection is precisely $\{1,2,\ldots,100\}$.

    Let $L_{e_1}, L_{e_2}, L_{e_3},\ldots$ be an enumeration of all the languages in $\mcC$ (with no repetition), where $L_{e_i} = \{1,2,\ldots,100\} \cup \Z_{\ge e_i}$ for $e_i \in \Z$. Let $t_1,t_2,t_3,\ldots$ be a sequence of candidate generation times for these languages, where every $t_i$ is finite and at least 1. We call the sequence \textit{admissible} if it satisfies the following property: for every $t \in \{1,2,\ldots,100\}$, there exists $n_t < \infty$ such that the set $\mcC_t = \{L_{e_i}: t_i \le t\}$ satisfies $e_i \le n_t$ for every $L_{e_i} \in \mcC_t$. We claim that there exists an algorithm that achieves the generation times given by the sequence $t_1,t_2,\ldots$ if and only if it is admissible. 

    For one direction, suppose that the sequence is not admissible, meaning that there exists some $t \in \{1,2,\ldots,100\}$ for which the set $\mcC_t = \{L_{e_i}:t_i \le t\}$ comprises of languages $L_{e_{i_1}}, L_{e_{i_2}}, \ldots$ where $\lim_{n \to \infty}e_{i_n}=\infty$. Assume for the sake of contradiction that there exists an algorithm that correctly achieves the generation times given in the sequence. Consider feeding $1,2,\ldots,t$ as input to the algorithm in the first $t$ time steps, and let $z_t$ be the most recent string generated by the algorithm. Then, by assumption, there exists a language $L_{e_n} \in \mcC_t$ for which $e_n > z_t$. We can now complete the input by enumerating the rest of the strings in $L_{e_n}$. Note that the algorithm generated a string outside $L_{e_n}$ at time step $t$, violating its claimed generation time of $t_i \le t$ for $L_{e_n}$.

    For the other direction, suppose that the sequence is admissible. Consider the algorithm that operates as follows: at any time step $t$, if $|S_t| \in \{1,2\ldots,100\}$, it obtains the subcollection $\mcC_t = \{L_{e_i}: t_i \le |S_t|\}$; otherwise if $|S_t| > 100$, the algorithm obtains $\mcC_t = \cap_{L \in \mcC, S_t \subseteq L}L$. If $\mcC_t$ is non-empty, and $\cap_{L \in \mcC_t}L$ is also non-empty, the algorithm generates a new string from $\cap_{L \in \mcC_t}L$ (or repeats a string if there are no new strings available in this set); otherwise, it generates an arbitrary string. Note that this algorithm is oblivious to its input.

    We claim that for every language $L_{e_i} \in \mcC$, this algorithm achieves the desired generation time of $t_i$. To see this, it suffices to argue that for every $t$ that satisfies $|S_t| \ge t_i$, $\mcC_t$ contains $L_{e_i}$, and that $\cap_{L \in \mcC_t}L$  is infinite. 
    
    First, consider any $t$ for which $t_i \le |S_t| \le 100$. By definition of the algorithm then, the subcollection $\mcC_t$ constructed by the algorithm contains $L_{e_i}$; in particular, it is non-empty. Furthermore, by admissibility, we have that there exists $n_t < \infty$ such that $e_j \le n_t$ for every $L_{e_j} \in \mcC_t$. In particular, this implies that every language in $\mcC_t$ contains all the numbers that are at least $n_t$, meaning that $\cap_{L \in \mcC_t}L$ is infinite. 

    Now, consider a $t$ for which $|S_t| > 100$; in this case, the algorithm constructs the subcollection $\mcC_t = \cap_{L \in \mcC, S_t \subseteq L}L$. Since we are considering $L_{e_i}$ to be the language being enumerated as input, $S_t \subseteq L_{e_i}$. Furthermore, since we argued at the beginning that the closure dimension of $\mcC$ is 100, we immediately have that $\cap_{L \in \mcC, S_t \subseteq L}L$ is infinite.

    To finish the proof, assume for the sake of contradiction that there exists an algorithm that achieves non-uniform generation times $t_1,t_2,\ldots$ for $L_{e_1}, L_{e_2},\ldots$ which are Pareto-optimal. By our reasoning above, the sequence $t_1,t_2,\ldots$ must be admissible. But now, consider any $t_i > 1$ in the sequence; such a $t_i$ must exist, since otherwise, if $t_i=1$ for every $i$, we can simply input, say 1, at the first time step, post which the algorithm must generate $z \neq 1$, and then the algorithm would have violated its guarantee that $t_{i}=1$ for $e_i=z$, if we were to complete enumerating $L_{e_i}$ beyond the first time step. So, for such a $t_i$ that satisfies $t_i > 1$, consider updating $t_i=1$, and let the other generation times in the sequence remain as is. This updated sequence strictly Pareto-dominates the earlier sequence. Recall also that $t_i$ corresponds to the generation time for the language $L_{e_i}$; hence, we observe that the updated sequence remains admissible. Namely, each previous value of $n_t$ for $t \in \{1,2,\dots,100\}$ may simply be updated to $n_t = \max(n_t, e_i) < \infty$. Therefore, since the updated sequence is admissible, there exists an algorithm that achieves the generation times given in the updated sequence. Consequently, the algorithm we started with could not have been Pareto-optimal.
\end{proof}

We conclude by arguing that both the collections from \Cref{claim:Pareto-impossibility-non-uniform} and \Cref{claim:Pareto-impossibility-uniform} do not satisfy the sufficient condition of Pareto-optimality in \Cref{thm:Pareto-optimality-guarantee}. To see this, for either of these collections $\mcC$, note that the intersection of languages in any finite subcollection is infinite; since the $m^\star(\cdot)$ values in the Pareto-optimal sequence computed in \Cref{subsec:Pareto-optimal-sequence} correspond to finite intersections of finite subcollections, we have that $m^\star(L_j)=0$ for every $L_j \in \mcC$ (in fact, this holds true irrespective of the original ordering of $\mcC$ that we start with to compute the $m^\star(\cdot)$ values). Thus, the condition does not hold for $t=1$, since $|\{j:m^\star(L_j)+1 \le 1\}|=\infty$. 

\section{Pareto-optimal Noisy Non-uniform Generation}
\label{appsec:Pareto-optimal-noisy}

We first claim that the witness set $T$ constructed in Step A of the while loop in Procedure \ref{proc:noisy} is always bounded.

\begin{claim}[Witness set is bounded]
    \label{claim:T-bounded-noisy}
    The set $T$ in Step A of the while loop always has bounded size. In particular, $m^\star_n(L_i) < \infty$ for every $n,i$.
\end{claim}
\begin{proof}
    Consider any subcollection $\mcC_{\chk}$ of $(L'_1,\ldots,L'_j)$ that contains $L'_j$, and has finite intersection. In particular, consider the size of the largest finite intersection, i.e.,
    \begin{align}
        \label{eqn:largest-T-noisy}
        d^\star = \max_{\substack{\mcC_{\chk} \subseteq (L'_1,\ldots,L'_j), \\ \mcC_{\chk} \ni L'_j, \\ |\cap_{L' \in \mcC_{\chk}}L'| < \infty }}|\cap_{L' \in \mcC_{\chk}}L'|.
    \end{align}
    Since $(L'_1,\ldots,L'_j)$ is finite, there are only finitely many $\mcC_{\chk}$ to consider, and each of these realize a finite intersection. Consequently, $d^\star$ is finite.

    Now, let $a^\star$ be the \textit{largest} noise level under consideration in the collection $\diag(\mcC_{n, i})$, i.e.,
    \begin{align*}
        a^\star = \max\{a: L_{a,b} \in \diag(\mcC_{n, i})\},
    \end{align*}
    and observe that $a^\star$ is finite (since $n',h$ are finite). Furthermore, since every $L' \in \mcC_{\chk}$ corresponds to some $L_{a,b} \in \diag(\mcC_{n, i})$, we have that for every $L_{a,b}$ corresponding to $L' \in \mcC_{\chk}$, it is the case that $a \le a^\star$. 

    So, suppose for the sake of contradiction that there exists arbitrarily large $T$, for which there exists a subcollection $\mcC_{\chk}$ satisfying the three conditions listed under Step A. In particular, this means that there exists such a $T$ of size $d > d^\star+l \cdot a^\star$ (where, recall that $l$ is the common size of $\mcC'_l$ and $\diag(\mcC_{n,i})$, and an upper bound on the size of $\mcC_{\chk}$). Since $T$ satisfies the second condition in Step A, we have that for every $L_{a,b}$ corresponding to $L' \in \mcC_{\chk}$, $T$ is $a$-contained in $L_b$, i.e., $L_{a,b}$ contains at least $|T|-a$ strings in $T$. But this means that there are at least $|T|-|\mcC_{\chk}| \cdot \max\{a\} \ge d -l \cdot a^\star > d^\star$ strings that are contained in \textit{every} language in $\mcC_{\chk}$. Additionally, $\mcC_{\chk}$ also satisfies the first and third conditions in Step A; that is, $L'_j \in \mcC_{\chk}$, and the intersection of languages in $\mcC_{\chk}$ is finite. Summarily, we have found $\mcC_{\chk}$ to be a subcollection of $(L'_1,\ldots,L'_j)$ that contains $L'_j$, \textit{and} has a finite intersection of size strictly larger than $d^\star$. This contradicts our derivation of $d^\star$ in \eqref{eqn:largest-T-noisy}. Thus, $T$ cannot be unbounded.
\end{proof}

Next, we show that the sequence of $m^\star_n(L_i)$ values constructed by the procedure forms a sequence of Pareto-optimal generation times. The claim is similar to \Cref{claim:non-uniform-generation-lb}.

\begin{claim}[$m^\star_n(\cdot)$ forms a Pareto-optimal sequence]
    \label{claim:noisy-non-uniform-gen-lb}
    Any algorithm $\mcG$ that satisfies $t_{n, \mcG}(L_i) < m^\star_n(L_i)+1$ for some $L_i$ must satisfy that $t_{n', \mcG}(L_j) > m^\star_{n'}(L_j)+1$ for some $L_{n', j}$ which is before $L_{n,i}$ in $\diag(\mcC_{n,i})$.
\end{claim}
\begin{proof}
    Note that the language $L_i$ must be such that $m^\star_n(L_i) > 0$, since $t_{n, \mcG}(\cdot)$ is always at least 1. In particular, this means that $\mcC(L_{n,i})$ is non-empty. $\mcC(L_{n,i})$ cannot also only contain copies of $L_i$, since otherwise, the intersection of languages in $\mcC(L_{n,i})$ would be infinite. So, $\mcC(L_{n,i})$ contains at least one language which is not equal to $L_i$. 
    
    Then, the adversary operates as follows: first, the adversary enumerates all the $m^\star_{n}(L_i)$ strings in $T(L_{n,i})$. Thereafter, the adversary enumerates all the remaining strings in the set $\left(\cap_{L \in \mcC(L_{n,i})}L\right)\setminus T(L_{n,i})$. Observe that at this point, the adversary has enumerated a sequence $S_t=\{x_1,\ldots,x_t\}$ where $t \ge m^\star_n(L_i)$, such that for every language $L \in \mcC(L_{n,i})$, which corresponds to some $L_{n',j}$, it is the case that $S_t$ is $n'$-contained in $L_j$. Let $z$ be the string generated by $\mcG$ at time $t$. Since $\mcG$ satisfies the guarantee that $t_{n, \mcG}(L_i) \le m^\star_n(L_i)$, it must be the case that $z \in L_i \setminus S_t$ (otherwise, the adversary can continue enumerating the rest of $L_i$ beyond this point, producing a valid enumeration of $L_i$ at noise level $n$, and causing $\mcG$ to violate its guarantee for $L_i$). But since we have exhausted out all the strings in the set $\cap_{L \in \mcC(L_{n,i})}L$, it also must be the case that $z \notin L_{n',j}$ for some $L_{n',j} \in \mcC(L_{n,i})$. The adversary can then continue enumerating the rest of $L_{n',j}$ beyond $t$, producing a valid enumeration of $L_j$ at noise level $n'$. Now, by construction of $\mcC(L_{n,i})$ and an argument similar to \Cref{observation:S-contains-smaller-complexity}, $m^\star_{n'}(L_j) < m^\star_n(L_i)$. So, if $t_{n', \mcG}(L_j)$ were to be at most $m^\star_{n'}(L_j)+1$, which is at most $m^\star_n(L_i)$, at time step $t$, the string generated by $\mcG$ should have belonged to $L_j$; we know this is not true, and hence $t_{n', \mcG}(L_j) > m^\star_{n'}(L_j)+1$.
\end{proof}

We now prove a claim similar to \Cref{claim:argmax-maintained}, which lends itself to the correctness of the stated noisy non-uniform generation algorithm.

\begin{claim}[Arg max maintained]
    \label{claim:argmax-maintained-noisy}
    For any $n \ge 0, i \ge 1$, consider any iteration $l$ of the nested for loop after the diagonal traversal has reached $L_{n,i}$ (i.e., $l \ge |\diag(\mcC_{n,i})|$), and let $k$ be the index of $L_{n,i}$ in $\mcC'_l=(L'_1,\ldots,L'_l)$ at the end of this iteration. Then,
    \begin{align}
        \label{eqn:argmax-maintained-noisy}
        T(L_{n,i}) \in \argmax_{T}\{|T|: T \text{ satisfies }(\star)\},
    \end{align}
    where $(\star)$ is the following condition: $T$ is finite, and there exists a subcollection $\mcC_{\chk} \subseteq (L'_1,\ldots,L'_k)$, such that $\mcC_{\chk} \ni L'_k$, the intersection of languages in $\mcC_{\chk}$ is finite, and for every $L' \in \mcC_{\chk}$, which corresponds to some $L_{a,b}$, $T$ is $a$-contained in $L_b$.
\end{claim}
\begin{proof}
    At the end of the iteration $l$ of the nested for loop when the diagonal traversal is \textit{at} $L_{n, i}$ (i.e., $l=|\diag(\mcC_{n,i})|$), the claim is true by definition of $T(L_{n,i})$. Now suppose as our inductive hypothesis that the claim is true for some iteration $l' \ge l$. We will show that it is true for $l'+1$. Suppose $k$ was the index of $L_{n,i}$ at the end of the ${l'}^\text{th}$ iteration of the nested for loop. If $k$ continues to be the index of $L_{n,i}$ at the end of the $(l'+1)^\text{th}$ iteration as well (meaning that the language after $L_{n,i}$ in the diagonal traversal---call it $L_{\next}$---was placed at some index larger than $k$ in $\mcC'_{l'+1}$), the claim continues to be true simply by the inductive hypothesis. Otherwise, $L_\next$ was placed at an index in $\{1,\ldots, k\}$, and so $L_{n,i}$ was moved to index $k+1$ in $\mcC'_{l'+1}$. Consider the precise iteration of the while loop when $j$ was equal to $k+1$ (so that $L'_{k+1}$ was $L_\next$), and $m_{\chk}$ was found to be at most $m^\star_n(L_i)$. At this point, Step C tells us to swap $L'_k$ and $L'_{k+1}$ in $\mcC'_{n+1}$ (so that $L'_{k+1}$ becomes equal to $L_{n,i}$, and $L'_k$ becomes equal to $L_\next$), and thereafter, note that the (unordered) subcollection of languages before $L'_{k+1}$ remains unchanged until the while loop breaks, and the iteration of the for loop completes. So, it suffices to argue that, \textit{before swapping}, $T(L_{n,i})$ belongs to the $\argmax$ of $|T|$ over the set $\{T\}$, where $T$ satisfies the condition: $T$ is finite, and there exists a subcollection $\mcC_\chk \subseteq (L'_1,\ldots,L'_k, L'_{k+1})$, such that $\mcC_\chk \ni L'_k$, the intersection of languages in $\mcC_\chk$ is finite, and for every $L' \in \mcC_\chk$, which corresponds to some $L_{a,b}$, $T$ is $a$-contained in $L_b$. Recall that our inductive hypothesis already guarantees that $T(L_{n,i})$ belongs to the arg max over the set of $\{T\}$, where the condition on $T$ only considers subcollections of $(L'_1,\ldots,L'_{k})$. That is, the only additional language now being considered in the condition is $L'_{k+1}$. So, if the max increased from what it was when $L'_{k+1}$ was not being considered, it \textit{must} be the case that for every $T$ belonging to the new arg max, the subcollection $\mcC_{\chk}$ for that $T$ contains $L'_{k+1}$. But this means that such a $T$---which has size strictly larger than $T(L_{n,i})$---satisfies the condition: $T$ is finite, and there exists a subcollection $\mcC_\chk \subseteq (L'_1,\ldots,L'_k, L'_{k+1})$, such that $\mcC_\chk \ni L'_{k+1}$, the intersection of languages in $\mcC_\chk$ is finite, and for every $L' \in \mcC_\chk$, which corresponds to some $L_{a,b}$, $T$ is $a$-contained in $L_b$. This contradicts that $m_{\chk}$ is at most $m^\star_n(L_i)$ (since $|T(L_{n,i})|=m^\star_n(L_{i})$). Thus, it must be the case that $T(L_{n,i})$ continues to be in the arg max. This completes the proof of the inductive step.
\end{proof}

We can now argue the correctness of the noisy generation algorithm stated in \Cref{sec:Pareto-optimal-noisy}, showing also that it is almost Pareto-optimal.

\generationtimesnoisy*

\begin{proof}
    Suppose that the adversary chose a noise level $n \ge 0$, and presented the algorithm with an enumeration of $L_i$ at noise level $n$, and suppose that we are at a time step $t$ for which $|S_t| \ge t^\star_n(L_i, \mcC)$. Observe that since $t \ge |S_t| \ge g(n,i)$, and $f$ is non-decreasing, $L_{n,i}$ is under consideration at $t$, and belongs to the collection $\mcC'_{f(t)}=(L'_1,\ldots,L'_{f(t)})$ constructed by the algorithm. Let $k$ be the index of the language $L_{n,i}$ in this collection. Note that $S_t$ is $n$-contained in $L_i$, since the adversary is in fact enumerating $L_i$ at noise level $n$. So, we only need to argue that when $L'_k=L_{n,i}$ is encountered by the algorithm as it traverses $\mcC'_{f(t)}$, it finds that $|\cap_{L \in I_t \cup \{L'_k\}}L|=\infty$ . Assume for the sake of contradiction that $|\cap_{L \in I_t \cup \{L'_k\}}L| < \infty$. In this case, observe that we have found a finite $S_t$ of size at least $m^\star_n(L_i)+1$, for which $I_t \cup \{L'_k\}$ is a subcollection of $(L'_1,\ldots,L'_k)$, that contains $L'_k$, has finite intersection, and for every $L$ in the subcollection, which corresponds to some $L_{a,b}$, $S_t$ is $a$-contained in $L_b$. But this implies that the maximum in the RHS of \eqref{eqn:argmax-maintained-noisy} in \Cref{claim:argmax-maintained-noisy} is at least $m^\star_n(L_i)+1$, which contradicts that $T(L_{n,i})$ realizes the maximum (since $|T(L_{n,i})|=m^\star_n(L_i)$). Thus, $L_{n,i}$ passes both the checks when it is encountered by the algorithm, and hence gets added to $I_t$.
\end{proof}

Finally, we derive a sufficient condition for when the algorithm is exactly Pareto-optimal.

\begin{theorem}[Sufficient Condition for Exact Pareto-optimality]
    \label{thm:Pareto-optimality-guarantee-noisy}
    Let $\mcC=(L_1,L_2,\ldots)$ be a collection that satisfies the following property: for every $t \in \N$, $|\{(n', j): m^\star_{n'}(L_j)+1\le t\}| < \infty$, where the $m^\star_n(\cdot)$ values are those computed in Procedure \ref{proc:noisy}. Then, there exists a noisy non-uniform generation algorithm that achieves a Pareto-optimal sequence of generation times for $\mcC$.
\end{theorem}
\begin{proof}
    Consider the function $f:\N \to \N$ defined as $f(t)=\max_{n',j}\{|\diag(\mcC_{n',j})|:m^\star_{n'}(L_j)+1 \le t\}$. Observe that $f$ is a non-decreasing function, and that $\lim_{t \to \infty}f(t)=\infty$. Furthermore, observe also that for any $n, i$, $f(m^\star_n(L_i)+1) \ge |\diag(\mcC_{n,i})|$. Therefore, if $g(n, i)$ is the smallest number $j$ for which $f(j) \ge |\diag(\mcC_{n,i})|$, then $g(n, i) \le m^\star_n(L_i)+1$. So, if we choose to run the algorithm given in \Cref{sec:Pareto-optimal-noisy} with this function $f$, the guarantee of \Cref{theorem:noisy-non-uniform-gen-ub} implies $t^\star_n(L_i)=m^\star_n(L_i)+1$ for every $L_i$ enumerated at noise level $n$. From \Cref{claim:noisy-non-uniform-gen-lb}, we conclude that the algorithm is Pareto-optimal.
\end{proof}

\section{Pareto-optimal Representative Non-uniform Generation}
\label{appsec:Pareto-optimal-representative}

We will first show that every $m^\star_\alpha(L_i)$ is finite.

\begin{claim}[Witness set is bounded]
    \label{claim:T-bounded-representative}
    The set $T$ in Step i of the while loop in Procedure \ref{proc:representative} always has bounded size. In particular, $m^\star_\alpha(L_i) < \infty$ for every $L_i$.
\end{claim}
\begin{proof}
    Consider any subcollection $\mcC_{\chk}$ of $(L'_1,\ldots,L'_j)$ that contains $L'_j$. 
    Let $T(\mcC_{\chk})$ be a finite subset of $\cap_{L \in \mcC_{\chk}}L$ which suffers group scarcity with respect to $\mcC_{\chk}$ and $\mcA$. 
    We will argue that $T(\mcC_{\chk})$ cannot be arbitrarily large. It suffices to reason about the case when $\cap_{L \in \mcC_{\chk}}L$ is infinite, since otherwise, $T(\mcC_{\chk})$ is immediately bounded. 

    To this end, consider the set of groups $R_1$ for which
    \begin{align}
        \label{eqn:R-1}
        R_1 &= \{g \in [K]: A_g \cap \left(\cap_{L \in \mcC_{\chk}}L \right)=\infty\}.
    \end{align}
    Note that since we assumed $\cap_{L \in \mcC_{\chk}}L$ is infinite, and $\mcA$ is a partition of $U$, $R_1$ is necessarily non-empty. On the other hand, let $R_2$ denote the groups for which
    \begin{align}
        \label{eqn:R-2}
        R_2 &= \{g \in [K]: A_g \cap \left(\cap_{L \in \mcC_{\chk}}L \right) <\infty\}.
    \end{align}
    In particular, let
    \begin{align}
        p = \max_{g \in R_2}|A_g \cap \left(\cap_{L \in \mcC_{\chk}}L \right)|.
    \end{align}
    Using that $K < \infty$, and the definition of the set $R_2$, we have that $p < \infty$. We can now claim that the size of $T(\mcC_{\chk})$ may be at most $\frac{Kp}{\alpha}$, which is a finite number. To see this, suppose $T(\mcC_{\chk})$ is a finite subset of $\cap_{L \in \mcC_{\chk}}L$ that has size strictly larger than $\frac{Kp}{\alpha}$. Observe that the scarce groups for $T(\mcC_{\chk})$ must necessarily be a subset of $R_2$. In that case, the largest that $\emp_{T(\mcC_{\chk})}^{\mcA}(g)$ can be is if we include all the strings in $A_g \cap \left(\cap_{L \in \mcC_{\chk}}L\right)$ in $T(\mcC_{\chk})$---but even in this case, we would have $\emp_{T(\mcC_{\chk})}^{\mcA}(g) < \frac{\alpha}{K} \le \alpha$. On the other hand, observe also that 
    \begin{align*}
        \sum_{g \in B}\emp_{T(\mcC_{\chk})}^{\mcA}(g) < \frac{\alpha|B|}{K} < \alpha < \alpha(K - |B|).
    \end{align*}
    Here, we used that $|B| < K$---this follows because we assumed that $\cap_{L \in \mcC_{\chk}}L$ is infinite, $T(\mcC_{\chk})$ is finite, and that $\mcA$ is a partition of $U$, which together imply that not all groups can be scarce. Thus, we can see that $T(\mcC_{\chk})$ would not satisfy group scarcity with respect to $\mcC_{\chk}$ and $\mcA$, if its size is larger than $\frac{Kp}{\alpha}$.

    Summarily, we have shown that for any $\mcC_{\chk}$ that contains $L'_j$, the largest that a finite subset $T(\mcC_{\chk})$ of $\cap_{L \in \mcC_{\chk}}L$ can be while suffering group scarcity with respect to $\mcC_{\chk}$ and $\mcA$ 
    is at most $\frac{Kp}{\alpha} := f(\mcC_{\chk}) < \infty$. So, define
    \begin{align}
        \label{eqn:d-star-def}
        d^\star = \max_{\substack{\mcC_{\chk} \subseteq (L'_1,\ldots,L'_j) \\ \mcC_{\chk} \ni L'_j}}f(\mcC_{\chk}).
    \end{align}
    Since the collection $(L'_1,\ldots,L'_j)$ is finite, there are only finitely many $\mcC_{\chk}$ to consider, and hence $d^\star$ is finite. 
    
    Finally, suppose for the sake of contradiction that there were to exist arbitrarily large finite $T$ that satisfy the conditions in Step i of the while loop. This would mean that there exists some $T$ of finite size $d > d^\star$, and a subcollection $\mcC_{\chk}$ of $(L'_1,\ldots,L'_j)$, such that $\mcC_{\chk}$ contains $L'_j$, $T \subseteq \cap_{L \in \mcC_{\chk}}L$, and also, $T$ suffers group scarcity with respect to $\mcC_{\chk}$ and $\mcA$.
    But this contradicts our derivation of \eqref{eqn:d-star-def} above, and concludes the proof that $T$ cannot be arbitrarily large.
\end{proof}

We can now argue that the sequence of $m^\star_\alpha(L_i)$ values forms a Pareto-optimal sequence.

\begin{claim}[$m^\star_\alpha(\cdot)$ forms a Pareto-optimal sequence]
    \label{claim:representative-non-uniform-gen-lb}
    Any algorithm $\mcG$ that satisfies $t_{\alpha, \mcG}(L_i) < m^\star_\alpha(L_i)+1$ for some $L_i$ must satisfy that $t_{\alpha, \mcG}(L_j) > m^\star_\alpha(L_j)+1$ for some $j < i$.
\end{claim}
\begin{proof}
    Note that the language $L_i$ must be such that $m^\star_\alpha(L_i) > 0$, since $t_{\alpha, \mcG}(\cdot)$ is always at least 1. In particular, this means that $\mcC(L_i)$ is non-empty.
    
    Suppose an adversary enumerates all the $m^\star_\alpha(L_i)$ strings in $T(L_i)$ in the first $m^\star_\alpha(L_i)$ time steps. Consider the distribution $\mcG_t$ output by the algorithm $\mcG$ at time $t=m^\star_\alpha(L_i)$. By assumption, $\mcG_t$ satisfies the properties required of $\alpha$-representation at $t$. There are two cases based on the mass that $\mcG_t$ assigns to the set $\cap_{L \in \mcC(L_i)}L \setminus S_t$:
    
    \paragraph{Case 1: $\Pr_{x \sim \mcG_t}\left[x \in \cap_{L \in \mcC(L_i)}L \setminus S_t \right]=1$.} The adversary can continue the enumeration of any $L \in \mcC(L_i)$ beyond $t$. Note that since $S_t=T(L_i)$, by definition of the set $T(L_i)$, we have that $S_t$ suffers group scarcity with respect to $\mcC(L_i)$ and $\mcA$. That is, if $B$ denotes the set of scarce groups with respect to $\mcC(L_i)$ and $S_t$, then either there exists a scarce group $g \in B$ that satisfies $\emp_{t}^{\mcA}(g) > \alpha$, or it is the case that $\sum_{g \in B}\emp_{t}^\mcA(g) > \alpha \cdot (K - |B|)$. Observe that $\mcG_t^{\mcA}(g)=\Pr_{x \sim \mcG_t}[x \in A_g]=0$ for every $g \in B$; otherwise, since $\mcG_t$ puts all its mass on $\cap_{L \in \mcC(L_i)}L\setminus S_t$ in the present case, we would have found a string in some $A_g \cap \left(\cap_{L \in \mcC(L_i)}L\setminus S_t\right)$, which contradicts that $g$ is a scarce group.
    
    So, suppose that there exists a scarce group $g \in B$ that satisfies $\emp_{t}^\mcA(g) > \alpha$. This implies that $|\emp_{t}^\mcA(g)-\mcG_t^\mcA(g)| = \emp_t^\mcA(g) > \alpha$, which contradicts the $\alpha$-representation property of $\mcG$. Thus, we cannot have $\emp_{t}^\mcA(g) > \alpha$ for any scarce group $g$.

    Now, suppose that $\sum_{g \in B}\emp_{t}^\mcA(g) > \alpha \cdot (K - |B|)$. This would imply that
    \begin{align*}
        \sum_{g \in [K]}\mcG_t^\mcA(g) &= \sum_{g \in [K] \setminus B}\mcG_t(g) \le \sum_{g \in [K] \setminus B}(\emp_t^{\mcA}(g)+\alpha) \\
        &= \sum_{g \in [K] \setminus B}\emp_t^\mcA(g) + \alpha \cdot (K - |B|) \\
        &< \sum_{g \in [K] \setminus B}\emp_t^\mcA(g) + \sum_{g \in B}\emp_t^\mcA(g) \\
        &= \sum_{g \in [K]}\emp_t^\mcA(g) = 1,
    \end{align*}
    which contradicts that $\mcG_t^\mcA$ is a valid probability distribution. Summarily, we deduce that Case 1 cannot happen.

    \paragraph{Case 2: $\Pr_{x \sim \mcG_t}\left[x \in \cap_{L \in \mcC(L_i)}L \setminus S_t \right]< 1$.} In this case, observe that there exists $L \in \mcC(L_i)$, such that $\Pr_{x \sim \mcG_t}[x \notin L \setminus S_t] > 0$. Moreover, $\mcC(L_i)$ cannot only contain copies of $L_i$; otherwise, the adversary can continue enumerating the rest of $L_i$ beyond $t=m^\star_\alpha(L_i)$, and $\mcG$ will have violated outputting a valid string at time $t$. So, it must be the case that $L \neq L_i$. The adversary can then continue enumerating this $L$ beyond $t$. Observe that by the manner in which $m^\star_\alpha(L_i)$ and $\mcC(L_i)$ are determined, such an $L$ corresponds to some $L_j$ for $j < i$ in the given specification of the collection $\mcC=(L_1,L_2,\ldots)$. Furthermore, by an argument identical to \Cref{observation:S-contains-smaller-complexity}, it is also the case that $m^\star_\alpha(L_j) < m^\star_\alpha(L_i)$. So, if it were the case that $t_{\alpha, \mcG}(L_j, \mcC)$ is at most $m^\star_\alpha(L_j)+1$, which is at most $m^\star_\alpha(L_i)$, then $\mcG_t$ should satisfy the guarantee that $\Pr_{x \sim \mcG_t}[x \in L_j \setminus S_t] = 1$; we just argued that this is not true, concluding the proof.
\end{proof}

The following claim shows again that whenever a language is moved forward when performing the insertion sort, its witness set continues to remain in the arg max of the relevant criterion in the prefix of its new position.

\begin{claim}[Arg max maintained]
    \label{claim:argmax-maintained-representative}
    For every $i\ge1$, and for every $t \ge i$, suppose $k$ is the index of language $L_i$ in $\mcC'_t=(L'_1,\ldots,L'_t)$ at the end of the $t^\text{th}$ iteration of the for loop. Then, $T(L_i)$ satisfies that
    \begin{align}
        \label{eqn:argmax-maintained-representative}
        T(L_i) \in \argmax_{T}\left\{|T|: T \text{ satisfies } (\star)\right\},
    \end{align}
    where $(\star)$ is the following condition: $T$ is finite, and there exists a subcollection $\mcC_{\chk} \subseteq (L'_1,\ldots,L'_k)$, %
    such that $\mcC_{\chk} \ni L'_k$, every $L \in \mcC_{\chk}$ contains $T$, and $T$ suffers group scarcity with respect to $\mcC_{\chk}$ and $\mcA$.
\end{claim}
\begin{proof}
    For $t=i$, the claim is true by definition of $T(L_i)$. Now suppose as our inductive hypothesis that the claim is true for some $t \ge i$. We will show that it is true for $t+1$. Suppose $k$ was the index of $L_i$ at the end of the $t^\text{th}$ iteration of the for loop. If $k$ continues to be the index of $L_i$ at the end of the $(t+1)^\text{th}$ iteration as well (that is, the language $L_{t+1}$ was placed at some index larger than $k$ in $\mcC'_{t+1}$), the claim continues to be true by the inductive hypothesis. Otherwise, the language $L_{t+1}$ was placed at an index in $\{1,\ldots, k\}$, and so $L_i$ was moved to index $k+1$ in $\mcC'_{t+1}$. Consider the precise iteration of the while loop when $j$ was equal to $k+1$ (so that $L'_{k+1}$ was $L_{t+1}$), and $m_{\chk}$ was found to be at most $m^\star_\alpha(L'_k)$. At this point, Step iii would tell us to swap $L'_{k}$ and $L'_{k+1}$ (so that $L'_{k+1}$ becomes equal to $L_i$, and $L'_k$ becomes equal to $L_{t+1}$), and thereafter, observe that the (unordered) subcollection of languages before $L'_{k+1}$ remains unchanged until the while loop breaks, and the iteration of the for loop completes. So, it suffices to argue that, \textit{before swapping}, $T(L_i)$ belongs to the arg max of $|T|$ over the set $\{T\}$, where $T$ satisfies the condition: $T$ is finite, and there exists a subcollection $\mcC_{\chk} \subseteq (L'_1,\ldots,L'_k,L'_{k+1})$ where $\mcC_{\chk} \ni L'_{k}$, every language $L \in \mcC_{\chk}$ contains $T$, and furthermore, $T$ suffers group scarcity with respect to $\mcC_{\chk}$ and $\mcA$.
    Recall that our inductive hypothesis already guarantees that $T(L_i)$ belongs to the arg max over the set $\{T\}$, where the condition on $T$ only considers subcollections of $(L'_1,\ldots,L'_{k})$. That is, the only additional language now being considered in the condition is $L'_{k+1}$. So, if by the introduction of this new language in the condition, the max were to increase, it must be the case that for every $T$ belonging to the new arg max, the subcollection \textit{necessarily} contains $L'_{k+1}$. But this means that such a $T$, which is finite and has size strictly larger than $T(L_i)$, satisfies that: there exists a subcollection $\mcC_{\chk} \subseteq (L'_{1},\ldots,L'_{k+1})$, $\mcC_{\chk} \ni L'_{k+1}$, every language $L \in \mcC_{\chk}$ contains $T$, and furthermore,  $T$ also suffers group scarcity with respect to $\mcC_{\chk}$ and $\mcA$. This contradicts $m_{\chk}$ being at most $m^\star_\alpha(L'_k)$ (since $m^\star_\alpha(L'_k)=m^\star_\alpha(L_i)=|T(L_i)|$). Thus, it must be the case that $T(L_i)$ continues to be in the arg max, completing the inductive step.
\end{proof}

\generationtimerepresentative*

\begin{proof}
    First, observe that for any $t$ satisfying $|S_t| \ge g(i)$, $L_i$ belongs to the collection $\mcC'_{f(t)}=(L'_1,\ldots,L'_{f(t)})$ under consideration by the algorithm. Suppose $k$ is the index for which $L'_k=L_i$. We only need to argue that $L'_k$ satisfies the required check. Assume for the sake of contradiction that $L'_k$ fails the check. In this case, $S_t$, which has size at least $m^\star_\alpha(L_i)+1$, satisfies that $I_t \cup \{L'_k\}$ is a subcollection of $(L'_1,\ldots,L'_k)$ which contains $L'_k$, all languages in the subcollection contain $S_t$, and furthermore, $S_t$ suffers group scarcity with respect to $I_t \cup \{L'_k\}$ and $\mcA$.
    Thus, the set $I_t$ constructed by the algorithm above contains $L_i$.

    We will now show that, for any enumeration of $L_i$, the output distribution $\mcG_t$ satisfies $\Pr_{x \sim \mcG_t}[x \in L_i \setminus S_t]=1$ for every $t$ satisfying $|S_t| \ge \max(g(i), m^\star_\alpha(L_i)+1)$. As argued %
    above, we know that the collection $I_t$ constructed by the algorithm at time $t$ contains $L_i$; in particular, it is non-empty. We also justified in the description of the algorithm that in the case that $I_t$ is non-empty, $\mcG_t$ has mass only on strings \textit{not} belonging to $S_t$, and belonging to \textit{every} language in $I_t$. Since $I_t$ contains $L_i$, we are done.

    We will now argue that the algorithm satisfies $\alpha$-representation for every $t \ge 1$. In the case that $I_t$ is found to be empty, observe that $\mcG_t$ is set to be $\emp_t$, and so, $\|\mcG_t^{\mcA}-\emp_t^\mcA\|_\infty=0$. 
    
    So, consider the case where $I_t$ is non-empty, and consider the set $B$ of scarce groups with respect to $I_t$ and $S_t$. Observe that by construction of $I_t$, $S_t$ does not suffer group scarcity with respect to $I_t$ and $\mcA$. This means that for every scarce group $g \in B$, $\emp_{t}^\mcA(g) \le \alpha$, and also, $\sum_{g \in B}\emp_{t}^\mcA(g) \le \alpha \cdot (K-|B|)$. Furthermore, when $I_t$ is non-empty, recall that the algorithm puts no mass on any string belonging to a scarce group. Thus, for every scarce group $g \in B$, we have that $|\mcG_t^\mcA(g)-\emp_t^\mcA(g)| = |\emp_t^\mcA(g)| \le \alpha$.

    Now, consider a non-scarce group $g \in [K] \setminus B$. Recall that the mass that the algorithm assigns to this group is entirely concentrated on the single string $s_g$ that it chooses. Furthermore, this mass is exactly equal to $\emp_{t}^\mcA(g)$, \textit{plus} an amount that is equal to $\frac{1}{K-|B|}\sum_{g \in B}\emp_t^\mcA(g)$, where the additional mass is the result of the algorithm distributing the mass that $\emp_t^\mcA$ has on scarce groups evenly among the non-scarce groups. So, we have that
    \begin{align*}
        |\mcG_t^\mcA(g)-\emp_t^\mcA(g)| &= \left|\frac{1}{K-|B|}\sum_{g \in B}\emp_t^\mcA(g)\right| \le \frac{\alpha \cdot (K-|B|)}{K-|B|} = \alpha.
    \end{align*}
    Thus, we have established that $\|\mcG_t^{\mcA}-\emp_t^\mcA\|_\infty \le \alpha$.
\end{proof}

We conclude by stating the corresponding sufficient condition for exact Pareto-optimality.
\begin{theorem}[Sufficient Condition for Exact Pareto-optimality]
    \label{thm:Pareto-optimality-guarantee-representative}
    Let $\mcC=(L_1,L_2,\ldots)$ be a collection that satisfies the following property: for every $t \in \N$, $|\{j: m^\star_\alpha(L_j)+1\le t\}| < \infty$, where the $m^\star_\alpha(\cdot)$ values are those computed in Procedure \ref{proc:representative}. Then, there exists an algorithm that achieves a Pareto-optimal sequence of non-uniform generation times for $\mcC$.
\end{theorem}
\begin{proof}
    Consider the function $f:\N \to \N$ defined as $f(t)=\max\{j:m^\star_\alpha(L_j)+1 \le t\}$. Observe that $f$ is a non-decreasing function, and that $\lim_{t \to \infty}f(t)=\infty$. Furthermore, observe also that for any $i$, $f(m^\star_\alpha(L_i)+1) \ge i$. Therefore, if $g(i)$ is the smallest number $j$ for which $f(j) \ge i$, then $g(i) \le m^\star_\alpha(L_i)+1$. So, if we choose to run the algorithm given in \Cref{sec:Pareto-optimal-representative} with this function $f$, the guarantee of \Cref{theorem:representative-non-uniform-gen-ub} implies $t^\star_\alpha(L_i)=m^\star_\alpha(L_i)+1$ for every $L_i$. From \Cref{claim:representative-non-uniform-gen-lb}, we conclude that the algorithm is Pareto-optimal.
\end{proof}

\end{document}